\documentclass[12pt]{article}

\usepackage[T1]{fontenc}
\usepackage[utf8]{inputenc}
\usepackage{natbib}
\usepackage{multibib}
\usepackage[english]{babel}
\usepackage{enumitem}
\usepackage{amsmath,amsthm,amssymb,bbm,bm}
\usepackage{mathtools}
\usepackage[usenames,dvipsnames,svgnames,table]{xcolor}
\definecolor{bestcolor}{gray}{0.5}
\usepackage[
    bookmarksdepth=3,
    linktoc=all,
    colorlinks=true,
    linkcolor=MidnightBlue!70,
    urlcolor=MidnightBlue!70,
    citecolor=MidnightBlue!70]{hyperref}

\newcites{Sup}{References}

\usepackage{setspace}
\usepackage{chngcntr}
\def\spacingset#1{\renewcommand{\baselinestretch}%
{#1}\small\normalsize} \spacingset{1}

\usepackage{graphicx}
\usepackage{booktabs}
\usepackage{tabularx}
\usepackage{algorithm}
\usepackage[noend]{algpseudocode}
\usepackage{tikz}
\usepackage{tikzscale}
\usetikzlibrary{bayesnet}

\usepackage[labelformat=simple]{subcaption}

\usepackage{newfloat}
\usepackage{framed}
\usepackage{float}
\floatstyle{boxed}
\newfloat{sampler}{thp}{los}
\floatname{sampler}{Algorithm}

\makeatletter
\let\c@algorithm\c@sampler
\makeatother

\newtheorem{definition}{Definition}
\newtheorem{theorem}{Theorem}
\newtheorem{lemma}{Lemma}
\newtheorem{corollary}{Corollary}

\newtheorem{proposition}{Proposition}
\newtheorem{assump}{Assumption}

\newtheorem{cond}{Condition}

\DeclareFloatingEnvironment[name=Algorithm]{myalgorithm}

\allowdisplaybreaks
\def\bal#1\eal{\begin{align}#1\end{align}}
\def\balnn#1\ealnn{\begin{align*}#1\end{align*}}

\DeclareMathOperator*{\argmin}{arg\,min}
\DeclareMathOperator*{\argmax}{arg\,max}

\DeclarePairedDelimiter\set{\{}{\}}
\DeclarePairedDelimiterX{\norm}[1]{\lVert}{\rVert}{#1}
\DeclarePairedDelimiterX{\abs}[1]{\lvert}{\rvert}{#1}
\DeclarePairedDelimiterX{\mean}[1]{\langle}{\rangle}{#1}
\DeclareMathOperator*{\spn}{span}
\DeclareMathOperator*{\tr}{tr}

\DeclareMathOperator*{\rank}{rank}

\DeclareMathOperator*{\diag}{diag}

\DeclareMathOperator*{\supp}{supp}

\DeclarePairedDelimiterX{\expectarg}[1]{[}{]}{%
    \ifnum\currentgrouptype=16 \else\begingroup\fi
    \activatebar#1
    \ifnum\currentgrouptype=16 \else\endgroup\fi
}

\DeclarePairedDelimiterX{\variancearg}[1]{(}{)}{%
    \ifnum\currentgrouptype=16 \else\begingroup\fi
    \activatebar#1
    \ifnum\currentgrouptype=16 \else\endgroup\fi
}

\newcommand{\innermid}{\nonscript\;\delimsize\vert\nonscript\;}
\newcommand{\activatebar}{%
    \begingroup\lccode`\~=`\|
    \lowercase{\endgroup\let~}\innermid
    \mathcode`|=\string"8000
}

\newcommand{\bbP}{\mathbb{P}}
\newcommand{\bbE}{\mathbb{E}}
\newcommand{\bbB}{\mathbb{B}}

\newcommand{\E}{\mathbb{E} \, \expectarg}

\newcommand{\iidsim}{\overset{\text{iid}}\sim}
\newcommand{\indsim}{\overset{\text{ind.}}\sim}

\newcommand{\SSIBP}{\text{SS-IBP}_d(a, \kappa, \bbP_{spike}, \bbP_{slab})}
\newcommand{\Exp}{\operatorname{Exponential}}
\newcommand{\Bern}{\operatorname{Bernoulli}}

\newcommand{\Pois}{\operatorname{Poisson}}
\newcommand{\Laplace}{\operatorname{Laplace}}
\newcommand{\Beta}{\operatorname{Beta}}

\def\mOne{{\mathbbm{1}}}
\newcommand{\ind}[1]{\mOne_{\{#1\}}}

\newcommand{\Reals}[1]{\mathbb{R}^{#1}}

\newcommand{\bY}{\mathbf{Y}}

\newcommand{\bX}{\mathbf{X}}
\newcommand{\bA}{\mathbf{A}}
\newcommand{\bB}{\mathbf{B}}
\newcommand{\bC}{\mathbf{C}}
\newcommand{\bD}{\mathbf{D}}
\newcommand{\bR}{\mathbf{R}}
\newcommand{\bS}{\mathbf{S}}
\newcommand{\bP}{\mathbf{P}}
\newcommand{\bQ}{\mathbf{Q}}

\newcommand{\bv}{\mathbf{v}}

\newcommand{\bbeta}{\boldsymbol\beta}

\newcommand{\boldeta}{\boldsymbol\eta}

\newcommand{\bTheta}{\boldsymbol\Theta}

\newcommand{\blambda}{\boldsymbol\lambda}
\newcommand{\bLambda}{\boldsymbol\Lambda}

\newcommand{\bpsi}{\boldsymbol\psi}

\newcommand{\bxi}{\boldsymbol\xi}

\newcommand{\bb}{\bm{\mathrm{b}}}

\newcommand{\bs}{\bm{\mathrm{s}}}
\newcommand{\bx}{\bm{\mathrm{x}}}

\newcommand{\bZ}{\bm{\mathrm{Z}}}
\newcommand{\bW}{\bm{\mathrm{W}}}
\newcommand{\bU}{\bm{\mathrm{U}}}
\newcommand{\bV}{\bm{\mathrm{V}}}
\newcommand{\bu}{\bm{\mathrm{u}}}

\newcommand{\bI}{\bm{\mathrm{I}}}
\newcommand{\bmu}{\bm{\mu}}

\title{A Spike-and-Slab Prior for Dimension Selection in Generalized Linear Network Eigenmodels}

\author{Joshua Daniel Loyal\thanks{Department of Statistics, Florida State University} \and Yuguo Chen\thanks{Department of Statistics, University of Illinois at Urbana-Champaign}}

\date{\empty}

\begin{document}

\maketitle

\begin{abstract}
    Latent space models (LSMs) are frequently used to model network data by embedding a network's nodes into a low-dimensional latent space; however, choosing the dimension of this space remains a challenge. To this end, we begin by formalizing a class of LSMs we call generalized linear network eigenmodels (GLNEMs) that can model various edge types (binary, ordinal, non-negative continuous) found in scientific applications. This model class subsumes the traditional eigenmodel by embedding it in a generalized linear model with an exponential dispersion family random component and fixes identifiability issues that hindered interpretability. Next, we propose a Bayesian approach to dimension selection for GLNEMs based on an ordered spike-and-slab prior that provides improved dimension estimation and satisfies several appealing theoretical properties. In particular, we show that the model's posterior concentrates on low-dimensional models near the truth. We demonstrate our approach's consistent dimension selection on simulated networks. Lastly, we use GLNEMs to study the effect of covariates on the formation of networks from biology, ecology, and economics and the existence of residual latent structure.

\end{abstract}


\noindent
{\it Keywords:} Generalized Linear Model; Model Selection; Network Latent Space Model; Statistical Network Analysis.

\onehalfspacing

\section{Introduction}

Networks data is at the center of modern statistical applications in various fields, including sociology, biology, ecology, and economics. A network describes the relations, or edges, between pairs of entities, or nodes. This information is represented as an $n \times n$ adjacency matrix $\bY$ with entries $\set{Y_{ij} \, : \, 1 \leq i,j \leq n}$ that quantify the relation between node $i$ and node $j$. These edge variables $Y_{ij}$ can be binary with values in $\set{0, 1}$ or weighted taking on any real value depending on the relation under study.  For example, world trade is often studied as a network with weighted edge variables, $Y_{ij}$, that represent the amount of bilateral trade in dollars between nations~\citep{ward2013, benedictis2014}.

A fundamental goal in the analysis of network data is to understand the relationship between the edge variables and additional dyadic covariates $\set{\bx_{ij} \in \Reals{p} \, : \, 1 \leq i,j \leq n}$.  A popular approach focuses on estimating the regression function $\mathbb{E}[Y_{ij} \mid \bx_{ij}]$ that relates the expected value of the edge variables to the dyadic covariates. Such a relationship is called {\it assortative} or {\it dissassortative} in the network literature depending on whether increasing a dyadic covariate increases or decreases the regression function. An example from economics is the gravity model, which models the expected bilateral trade between nations as a log-linear function of the nations' gross domestic product and geographic distance~\citep{tinbergen1962, anderson1979}, cf. Section~\ref{sec:applications}. However, the assumption made by these traditional regression methods that the edge variables are independent conditional on the covariates is almost always invalidated by the presence of strong network dependencies such as degree, transitivity, and clustering effects~\citep{kolaczyk2017, hoff2021}. 

In the statistics literature, a solution is to introduce additional latent variables into the regression to capture the residual network dependencies, such that conditional on the latent variables the edge independence assumption holds. These latent variables usually take the form of latent variable network models including the $\beta$-model~\citep{graham2017, yan2019, stein2022}, stochastic block models~\citep{mariadassou2010, huang2022}, and latent space models (LSMs)~\citep{hoff2002, handcock2007, krivitsky2008, ma2020}. For a comprehensive overview of network models, see \citet{goldenberg2010} and \citet{kolaczyk2017}. In this work, we assume the latent variables take the form of a LSM
which
has garnered popularity due to its ability to capture various network structures while remaining amiable to interpretation and visualization. The key idea behind LSMs is that each node $i$ can be represented by a latent position $\bu_i \in \Reals{d}$ in a low-dimensional ($d \ll n$) latent space, and nodes with similar latent positions will have a larger edge variable in expectation. 


This paper tackles two problems in the Bayesian analysis of networks using LSMs:
(1) Selecting the dimension of the latent space $d$ while also assessing its uncertainty, and (2) Accounting for this selection when assessing the covariate effects' statistical uncertainty. The most popular approach to dimension selection minimizes a model selection criterion~\citep{spiegelhalter2002, watanabe2010, sosa2021} or a metric calculated on data held out by a data splitting strategy such as cross-validation~\citep{hoff2008, chen2018, li2020}. For Bayesian LSMs, significant issues with these approaches are that they either are computationally intensive, lack theoretical support, or only apply to networks without covariates. As such, many Bayesian analyses forgo dimension selection altogether. Furthermore, any post-selection inference about the covariate effects will not account for the uncertainty induced by the dimension selection process.


Alternatively, a fully Bayesian solution places a prior over the latent space dimension and selects the dimension based on its posterior distribution. A difficulty with this approach is that most LSM likelihoods are invariant to the re-ordering of the latent space dimensions, so special priors are required to account for this non-identifiability. Early attempts used ordered shrinkage priors~\citep{bhattacharya2011, durante2014}; however, these priors often induce undesirable overshrinkage due to their sensitivity to hyperparameter choice~\citep{durante2017letters} and are unable to fully eliminate unnecessary dimensions from the model. As such, recent applications of the eigenmodel~\citep{guhaniyogi2020, guha2021} have employed spike-and-slab priors~\citep{ishwaran2005} that penalize increasing model size. However, these priors share many weaknesses with the original ordered shrinkage priors because their increasing penalization only holds in expectation. Furthermore, all existing priors lack any theoretical guarantees that their penalization penetrates through to the posterior.

This work addresses these problems for a large class of models we call \textit{generalized linear network eigenmodels} (GLNEMs), which are variants of the eigenmodel proposed by \citet{hoff2008, hoff2009}. Traditionally, the eigenmodel relates a latent Gaussian edge variable to the regression function through a link function. The use of latent Gaussian random variables permits Bayesian inference over a variety of edge-types by facilitating a flexible Metropolis-within-Gibbs Markov chain Monte Carlo (MCMC) algorithm; however, it hinders the interpretation of the covariate effects since one must interpret them conditional on these latent variables. Instead, GLNEMs directly embed the eigenmodel in a generalized linear modeling framework with an exponential dispersion family random component and possibly non-canonical link functions. Importantly, under the GLNEM framework, we introduce new identifiability constraints on the latent positions that allow one to interpret the covariate effects as marginal effects independent of the latent variables.

The main thrust of this article is then the development of a computationally efficient and theoretically grounded Bayesian approach that can automatically identify the dimension of the latent space for GLNEMs. Our approach to dimension selection begins with a prior that induces posterior zeros with high probability. 
We introduce a novel spike-and-slab prior coupled with a non-homogeneous Indian buffet process prior~\citep{griffiths2011} that increasingly pulls coefficients to the spike at zero, which we call the {\it truncated non-homogeneous spike-and-slab Indian buffet process}. The advantages of this prior are three-fold: (1) we demonstrate that the prior induces a stochastic ordering that alleviates issues arising from the GLNEM likelihood's invariance to the re-ordering of the latent space dimensions, (2) we provide a tail bound on the expected posterior latent space dimension that demonstrates concentration near the true number of dimensions, and (3) we identify hyperparameter values that are justified by our theoretical results. To the best of our knowledge, this is the first such theoretical guarantee for Bayesian LSMs.

Our final contribution addresses the challenge of designing an efficient MCMC
algorithm that applies for all GLNEMs. This task's difficulty is compounded by the fact that our proposed identifiability constraints require some parameters to lie on a sub-manifold of the Stiefel manifold. We address this issue by introducing a new parameter expanded prior based on the QR decomposition that is amiable to gradient based sampling methods. Under this prior, the target posterior distribution has marginal densities for the parameters of interest that satisfy the identifiability constraints. Furthermore, we use this prior to design a simple Metropolis-within-Gibbs sampler with efficient Hamiltonian Monte Carlo proposals~\citep{neal2011, hoffman2014} that applies for any GLNEM.

The article is structured as follows. Section~\ref{sec:model} introduces GLNEMs. 
Section~\ref{sec:dim_select} develops the truncated non-homogeneous spike-and-slab Indian buffet process and demonstrates properties sufficient for dimension selection. Section~\ref{sec:theory} provides a tail bound on the GLNEM's posterior latent space dimension under the proposed prior. Section~\ref{ss:sec:estimation} describes our parameter expansion strategy and MCMC sampler. Section~\ref{ss:sec:simulation} investigates the proposed methods in a simulation study, and Section~\ref{sec:applications} applies them to networks from biology, ecology, and economics. Lastly, Section~\ref{sec:discussion} contains a discussion. 


\section{Generalized Linear Network Eigenmodels}\label{sec:model}

Here, we outline a class of latent space models for network data with dyadic covariates we call generalized linear network eigenmodels (GLNEMs). We start by fixing notation. Let $B_{1:L}$ refer to the set $\set{B_1, \dots, B_L}$ where $B_{\ell}$ is any indexed object. We denote the $n$-dimensional vector of ones and zeros by $\mathbf{1}_n$ and $\mathbf{0}_n$, respectively. We let $\bI_d$ represent the $d \times d$ identity matrix. For a vector $\bv \in \Reals{d}$, we use $\diag(\bv)$ to denote the $d \times d$ diagonal matrix with the elements of $\bv$ on its diagonal. For a matrix $\bC \in \Reals{n \times m}$, we denote the Frobenius norm as $\norm{\bC}_F = (\sum_{ij} C_{ij}^2)^{1/2}$ and the number of non-zero elements as $\norm{\bC}_0$. We denote the Stiefel manifold, which contains $n \times d$ semi-orthogonal matrices, as $\mathcal{V}_{d,n} = \set{\bU \in \Reals{n \times d} \, : \, \bU^{\top} \bU = \bI_d}$. We use $\indsim$ and $\iidsim$ to denote independently distributed and independent and identically distributed, respectively. Lastly, the proofs of all subsequent results can be found in Section~\ref{sec:proofs} of the supplement.

\subsection{Model Description}


Suppose we observe an undirected network on a set of $n$ nodes, which is represented by an $n \times n$ adjacency matrix $\bY$ with symmetric real-valued entries $Y_{ij} = Y_{ji}$. For simplicity, we allow self-loops (i.e., diagonal entries) although these entries may be unobserved. We also observe additional $p$-dimensional dyadic covariates for each pair of nodes $i$ and $j$, denoted by $\bx_{ij} = (x_{ij,1}, \dots, x_{ij, p})^{\top}$. Let $\bX_k$ denote the $n \times n$ matrix with the $k$-th covariate of each dyad as entries, i.e., $(\bX_k)_{ij} = x_{ij, k}$ for $k = 1, \dots, p$. 

In the tradition of generalized linear models (GLMs), an intuitive statistical model for the adjacency matrix specifies a systematic component that describes $\bY$'s expected value and a random component that categorizes its distribution. 
GLNEMs relate the adjacency matrix's expected value to a linear function of the covariates and the latent space via a strictly increasing link function $g$ as follows:
\begin{equation}\label{eq:systematic}
    g(\mathbb{E}[\bY \mid \bX_{1:p}]) = \sum_{k=1}^p \bX_k \beta_k + \bU \bLambda \bU^{\top}, 
\end{equation}
where $\bbeta = (\beta_1, \dots, \beta_p)^{\top} \in \Reals{p}$, $\bLambda = \diag(\blambda) =\diag(\lambda_1, \dots, \lambda_d) \in \Reals{d \times d}$ is a real-valued diagonal matrix, $\bU = (\bu_1, \dots, \bu_n)^{\top} \in \Reals{n \times d}$ is a matrix with the nodes' latent positions as rows, and in a slight abuse of notation, $g(\bbE[\bY \mid \bX_{1:p}])$ is the link function applied element-wise to the matrix $\bbE[\bY \mid \bX_{1:p}]$. The $k$-th regression coefficient $\beta_k$ measures the extent of assortativity or disassortativity in the network attributed to the $k$-th covariate depending on whether $\beta_k$ is positive or negative. The $\bU\bLambda\bU^{\top}$ term captures latent low-rank residual network effects. This term takes the form of an eigen-decomposition, which motivates the eigenmodel name~\citep{hoff2008}. Similar to the regression coefficients, $\lambda_h$ measures the amount of assortativity or disassortativity associated with the $h$-th latent space dimension. 



For the random component, we assume the edge variables are independent conditional on the dyad-wise covariates and the latent space:
\begin{equation}\label{eq:random}
    Y_{ij} = Y_{ji} \mid \bx_{ij} \indsim Q\left\{\cdot \mid g^{-1}\left(\bbeta^{\top} \bx_{ij} + \left[\bU \bLambda\bU^{\top}\right]_{ij}\right), \phi\right\}, \qquad 1 \leq i \leq j \leq n,
\end{equation}
where $Q(\cdot \mid \mu, \phi)$ is a member of the exponential dispersion family with mean $\mu$ and known dispersion factor $\phi$. That is, we assume that the edge variables are independently drawn from densities of the form
\begin{equation}\label{eq:edm}
      q(y_{ij}; \theta_{ij}, \phi) = q_{ij}(y_{ij}) = \exp\left\{\frac{y_{ij} \, \theta_{ij} - b(\theta_{ij})}{\phi} + k(y_{ij}, \phi) \right\}, \qquad 1 \leq i \leq j \leq n,
\end{equation}
where $\theta_{ij}$ is a natural parameter lying in $\Theta \subset \Reals{}$, $\phi$ is a known dispersion factor, and $b$ and $k$ are known functions. Following the usual conventions for GLMs, we assume that the function $b$ is twice differentiable and strictly convex on $\Theta$ so that the second derivate $b''$ satisfies $b''(x) > 0$ for every $x \in \Theta$. It is easy to see that the expected value and variance of $Y_{ij}$ are $b'(\theta_{ij})$ and $\phi b''(\theta_{ij})$, respectively. The exponential dispersion family includes the Gaussian, binomial, Poisson, negative binomial, gamma, Tweedie, and beta distributions.  Based on our choice of systematic component, we have
\begin{equation}\label{eq:linpred}
    g(b'(\theta_{ij})) = g(\mathbb{E}[Y_{ij} \mid \bx_{ij}]) = \bbeta^{\top}\bx_{ij} + \left[\bU \bLambda \bU^{\top}\right]_{ij} = \bbeta^{\top}\bx_{ij} + \bu_i^{\top}\bLambda \bu_j.
\end{equation}
We refer to model (\ref{eq:systematic}) -- (\ref{eq:edm}) as a generalized linear network eigenmodel (GLNEM).

It is worth comparing the GLNEM framework to existing approaches. Unlike the eigenmodel proposed by \citet{hoff2008, hoff2009}, which assumes $g(\mathbb{E}[Y_{ij} \mid \bx_{ij}]) = \bbeta^{\top}\bx_{ij} + \bu_i^{\top}\bLambda\bu_j + \varepsilon_{ij}$, where $\varepsilon_{ij} \iidsim N(0, \sigma^2)$, GLNEMs do not use latent Gaussian random variables.
GLNEMs are also closely related to the generalized linear modeling framework for network data proposed by \citet{wu2017}, who replace $\bU\bLambda\bU^{\top}$ with a general low-rank matrix. Unlike the model of \citet{wu2017}, GLNEMs allow for both dispersion and non-canonical link functions, which can provide a better fit to real-world networks as shown in Section~\ref{sec:applications}. Next we present some concrete examples to demonstrate the utility of the GLNEM framework.

{\bf Example 1} (Binary Edges). For binary edge variables, an appropriate random component $Q$ is a Bernoulli distribution. 
Setting $g(x) = \Phi(x)$, the cumulative distribution function of the standard normal distribution, one recovers the eigenmodel in \citet{hoff2008}. Another common choice is the logistic link function. However, the GLNEM framework also includes the complementary log-log and log-log link functions, which have not appeared in the network literature. 

{\bf Example 2} (Ordinal Edges).  Ordinal edge variables are a common edge type, e.g., networks of counts. In this case, common choices for $Q$ are the Poisson or negative binomial distribution. Often $g(x) = \log(x)$ is used for both distributions despite being non-canonical for the latter. Furthermore, we found that the negative binomial model often results in lower-dimensional models when over-dispersion is present in the data, which often occurs in networks due to zero-inflation. See Section~\ref{sec:zero-inflation} of the supplement for details.   

{\bf Example 3} (Non-Negative Continuous Edges). Edge variables are often non-negative continuous so that $Y_{ij} \in \set{0} \cup \Reals{}_{+}$, e.g., world-trade networks. A solution proposed in the literature is the tobit model based on a latent Gaussian variable~\citep{sewell2016weights}; however, this model can be inappropriate for certain networks. For example, economists have scrutinized using the tobit model to model world trade due to the results depending on the units of measurement~\citep{silva2006}, e.g., trade measured in dollars or billions of dollars. A solution under the GLNEM framework that has not appeared in the network literature is the Tweedie distribution~\citep{jorgensen1987} with power parameter $1 < \xi < 2$, which is invariant to the edge variables' units. Furthermore, the Tweedie distribution is equivalent to a compound Poisson-gamma distribution, which is a reasonable data generating process for trade. See Section~\ref{sec:tweedie} of the supplement for more details.


\subsection{Parameter Identifiability and Marginal Effect Interpretation}

Identifiability of model (\ref{eq:systematic}) -- (\ref{eq:edm}) requires additional constraints on the model parameters. For $\bbeta$ and $\bU \bLambda\bU^{\top}$ to be individually identifiable, we place a sum-to-zero constraint on $\bU$'s columns and require the node-averaged covariate matrix $\bar{\bX} = (1/n)\sum_{j=1}^n (\bx_{1j}, \dots, \bx_{nj})^{\top} = (\bar{\bx}_1, \dots, \bar{\bx}_n)^{\top}$ be non-singular. For the magnitude of $\bLambda$ to be identifiable, we must constrain the scale of $\bU$. We adopt a constraint imposed by \citet{hoff2009} that requires $\bU$ to be a semi-orthogonal matrix. In summary, we make the following two assumptions: 

\begin{assump}\label{id:a1}
$\bU \in \bar{\mathcal{V}}_{d,n}$, where $\bar{\mathcal{V}}_{d,n} = \set{\bU \in \Reals{n \times d} \ \mid \ \bU^{\top}\bU = \bI_d \text{ and } \bU^{\top}\mathbf{1}_n = \mathbf{0}_d}$.
\end{assump}
\begin{assump}\label{id:a2}
$\rank(\bar{\bX}) = p$.
\end{assump}
\noindent Formally, we have the following proposition about the identifiability of the covariate effects.

\begin{proposition}[Identifiability of Covariate Effects]\label{prop:identify}
    Assume $\bY$ is drawn from model (\ref{eq:systematic}) -- (\ref{eq:edm}) with parameters $\set{\bbeta, \bU, \bLambda, \phi}$ that satisfy Assumptions \ref{id:a1} and \ref{id:a2}, then $\bbeta$ and $\bU \bLambda \bU^{\top}$ are identifiable. That is, if the probability distribution induced by two different parameterizations $\set{\bbeta, \bU, \bLambda, \phi}$ and $\set{\tilde{\bbeta}, \tilde{\bU}, \tilde{\bLambda}, \tilde{\phi}}$ coincide, then $\bbeta = \tilde{\bbeta}$ and $\bU \bLambda \bU^{\top} = \tilde{\bU} \tilde{\bLambda} \tilde{\bU}^{\top}$.
\end{proposition}

Importantly, the sum-to-zero constraint allows one to interpret $\bbeta$'s elements as marginal effects, which is usually not the case for LSMs~\citep{minhas2019}. 
Specifically, often one wants to quantify how the covariates affect the edge variables associated with a single node $i$. It is natural to consider the average $n^{-1} \sum_{j=1}^n g(\mathbb{E}[Y_{ij} \mid \bx_{ij}]) = \bbeta^{\top} \bar{\bx}_{i}$. In other words, the elements of $\bbeta$ represent the effect a covariate has on the average value of $g(\bbE[Y_{ij} \mid \bx_{ij}])$ after marginalizing over the latent positions. For example, in the context of a Bernoulli GLNEM with a logistic link, $\beta_k$ is the marginal additive effect on the average log-odds ratio  of forming an edge with node $i$ for a unit increase in feature $k$, i.e., a unit increase in the average covariate $\bar{x}_{i,k}$ increases the average log-odds of forming an edge with node $i$ by $\beta_k$. 
Similarly, under a log link we can interpret $e^{\beta_k}$ as the marginal multiplicative effect on the geometric mean of the expected edge variables involving node $i$. Note that we call these marginal effects because they are not conditional on keeping the latent positions fixed.

\section{A Spike-and-Slab Prior for Dimension Selection}\label{sec:dim_select}

Now, we turn to the main contribution of this paper which is the development of a Bayesian procedure that can automatically identify the dimension of the latent space.  In the context of GLNEMs, the dimension selection problem is equivalent to estimating $\rank(\bU \bLambda \bU^{\top}) = \norm{\bLambda}_0 \leq d$, i.e., the number of non-zero diagonal elements of $\bLambda$. As such our approach begins with a spike-and-slab prior on the diagonal elements of $\bLambda$ that induces posterior zeros with high probability. Furthermore, we expect additional dimensions to play a progressively less important role in describing the network structure. As a solution, we introduce a novel ordered spike-and-slab prior that induces a stochastic ordering of $\bLambda$'s elements and results in a posterior that asymptotically concentrates around the true dimension $d_0$.

\subsection{The Non-Homogeneous Spike-and-Slab Indian Buffet Process}

To construct a prior with ordered shrinkage and dimension selection, we build upon the spike-and-slab priors~\citep{ishwaran2005} and the Indian buffet process~\citep{griffiths2011}. Specifically, let $\set{\eta_h}_{h=1}^d$ be a collection of $d$ random variables. We consider the following spike-and-slab prior on $\set{\eta_h}_{h=1}^d$ with a shrinking sequence of slab probabilities $\theta_1 > \theta_2 > \dots > \theta_d$ :
\begin{equation}\label{eq:ssibp_prior}
\begin{split}
&\eta_h \mid \theta_h \indsim \theta_h \mathbb{P}_{slab} + (1 - \theta_h) \mathbb{P}_{spike}, \qquad \theta_h = \prod_{\ell=1}^{h} \nu_{\ell}, \qquad h = 1, \dots, d, \\
&\nu_{1} \indsim \Beta(a, \kappa + 1), \qquad \nu_{\ell} \iidsim \Beta(a, 1), \qquad \ell = 2, \dots, d,
\end{split}
\end{equation}
where $a > 0$, $\kappa \geq 0$, and $\bbP_{spike}$ and $\bbP_{slab}$ are the distributions of the spike and slab components, respectively. Because our prior uses the stick-breaking construction of the Indian buffet process, we denote the prior in 
(\ref{eq:ssibp_prior}) as the non-homogeneous spike-and-slab Indian buffet process prior truncated at $d$, or $\SSIBP$. The hyperparameters of this prior are determined by $\bbP_{slab}$, $\bbP_{spike}$, and the scalars $a $ and $\kappa$. We call this prior non-homogeneous because we draw $\nu_1$ from a two-parameter Beta distribution and the remaining $\nu_{2:d}$ parameters from a one-parameter Beta distribution. We will show that this non-homogeneous IBP is sufficient to obtain posterior concentration around the true dimension when used as a prior for $\bLambda$ in GLNEMs.

The ordering of the slab probabilities $\set{\theta_h}_{h=1}^d$ induces a stochastic ordering of the parameters $\set{\eta_h}_{h=1}^d$ under the prior, which is formalized in the following proposition.
\begin{proposition}[SS-IBPs Induce a Stochastic Ordering]\label{prop:ssibp-ordering}
For $\epsilon > 0$ and fixed $\eta_0 \in \Reals{}$, let $\mathbb{B}_{\epsilon}(\eta_0) = \set{\eta \, : \, \abs{\eta - \eta_0} \leq \epsilon}$ denote the $\epsilon$-ball centered at $\eta_0$. Under the $\SSIBP$, $\mathbb{P}(\abs{\eta_h - \eta_0} \leq \epsilon) < \mathbb{P}(\abs{\eta_{h+1} - \eta_0} \leq \epsilon)$ whenever $\mathbb{P}_{slab}(\mathbb{B}_{\epsilon}(\eta_0)) < \mathbb{P}_{spike}(\mathbb{B}_{\epsilon}(\eta_0))$.
\end{proposition}
Intuitively, for any $\eta_0$ that is favored by the spike distribution $\bbP_{spike}$ relative to the slab $\bbP_{slab}$, the $\SSIBP$ places more mass around $\eta_0$ as $h$ increases. We usually set $\eta_0 = 0$, which implies the prior places more mass around zero as $h$ increases.

Next we turn to an important property of the $\SSIBP$ prior's tail behavior that is crucial for obtaining a posterior that is adaptive to the latent space dimension and may be of interest for applications outside GLNEMs. When $\bbP_{spike} = \delta_0$ is a point mass concentrated at zero, the $\SSIBP$ prior induces an implicit prior on the number of non-zero elements of $\boldeta = (\eta_1, \dots, \eta_d)^{\top}$, that is, $\norm{\boldeta}_0$, that favors low-dimensions. Formally, the next theorem shows that this induced prior has exponentially decaying tails when we set $\bbP_{spike} = \delta_0$ and $a$ and $\kappa$ appropriately. 
\begin{theorem}[The Prior on $\norm{\boldeta}_0$ has Exponentially Decaying Tails]\label{thm:exp_decay}
If $\boldeta \sim \text{SS-IBP}_d(\allowbreak a, d^{1 + \delta}, \bbP_{spike}, \bbP_{slab})$ for $a \in (0, 1)$, $\delta \geq 6 / \log d$, and $\bbP_{spike} = \delta_0$, then for any $t \geq 1$,
\begin{equation}\label{eq:exp_decay}
    \bbP\left(\norm{\boldeta}_0 > t\right) \leq 2 \, e^{-t (\delta/6) \log d}.
\end{equation}
\end{theorem}
It has been widely observed in the literature on posterior asymptotics that an exponentially decaying tail is an sufficient property for obtaining optimally behaving posteriors~\citep{castillo2012, rockova2016}. Note that Theorem~\ref{thm:exp_decay} leaves the slab distribution unspecified; however, a sufficiently heavy-tailed slab is used in Theorem~\ref{thm:post_con} to ensure property (\ref{eq:exp_decay}) transfers to the posterior. 

\subsection[morestuff]{An SS-IBP Prior for $\bLambda$ in GLNEMs}

Now, we use the $\SSIBP$ to construct a prior for $\bLambda = \diag(\blambda)$ that promotes sparsity and induces a soft identifiability constraint that counteracts the permutation invariance of the GLNEM's likelihood. In particular, we use the following prior:
\begin{align}\label{eq:static_ssibp}
    &(\lambda_1, \dots, \lambda_d) \sim \SSIBP, \nonumber \\
    &\bbP_{slab} = \Laplace(b), \qquad \bbP_{spike} = \Laplace(v_0b),
\end{align}
where $\Laplace(b)$ denotes a Laplace distribution centered at the origin with scale $b$.

The next corollary to Proposition~\ref{prop:ssibp-ordering} states that this prior induces a stochastic ordering on $\bLambda$'s diagonal elements, which counteracts the likelihood's column permutation invariance. 
\begin{corollary}\label{prop:static_marginal_order}
    If $(\lambda_1, \dots, \lambda_h)$ are drawn from the SS-IBP prior defined in 
    (\ref{eq:static_ssibp}), then, for any $\epsilon > 0$, $\bbP(\abs{\lambda_{h}} \leq \epsilon) < \bbP(\abs{\lambda_{h+1}} \leq \epsilon)$, whenever $0 \leq v_0 < 1$.
\end{corollary}
This corollary follows immediately from the fact that the Laplace distribution puts greater mass around zero when the scale parameter is near zero. In particular, the condition $v_0 < 1$ ensures that the spike distribution is more concentrated around zero than the slab distribution.  In the remainder of this article, we will set $v_0 = 0$, which corresponds to using a discrete spike-and-slab prior with $\bbP_{spike} = \delta_0$, a point mass concentrated at 0. We provide the more general result because setting $v_0 > 0$ can have computational advantages.

\section{Theoretical Results on Dimension Selection}\label{sec:theory}

Next, we provide theoretical support for using the posterior distribution of $\norm{\bLambda}_0$ to infer the dimension of the latent space for GLNEMs when an SS-IBP prior is placed on $\bLambda$. Specifically, we consider the following model augmented with priors for the parameters:
\begin{align*}
    &Y_{ij} = Y_{ji} \mid \bx_{ij} \indsim Q\left\{\cdot \mid g^{-1}\left(\bbeta^{\top} \bx_{ij} + \left[\bU \bLambda\bU^{\top}\right]_{ij}\right), \phi\right\}, \qquad 1 \leq i \leq j \leq n,  \\
    &\bbeta \sim N(\mathbf{0}_p, \sigma^2_{\beta} \bI_p), \qquad \bU \sim \Pi_{\bU}, \qquad (\lambda_1, \dots, \lambda_d) \sim \SSIBP,
\end{align*}
where the edge variable's distribution function $Q$ is a member of the exponential dispersion family defined in 
(\ref{eq:random}), $\Pi_{\bU}$ is any distribution fully supported on $\mathcal{V}_{d,n}$, $\bbP_{spike} = \delta_0$, and $\bbP_{slab} = \Laplace(b)$. Because we will study the asymptotics of the posterior distribution as the number of nodes $n$ goes to infinity, we will denote the adjacency matrix by $\bY^{(n)}$ and the dyad-wise covariate matrices by $\bX_k^{(n)}$ for $1 \leq k \leq p$ to make the dependence on $n$ apparent. We denote the posterior distribution as $\Pi(\cdot \mid \bY^{(n)})$.



The goal of our analysis is to study the frequentist properties of the posterior distribution assuming the matrices $\bY^{(n)}$ are generated from model (\ref{eq:systematic}) -- (\ref{eq:edm}) with true non-zero latent space dimension $d_0$ and true parameters $\set{\bbeta_0, \bU_0, \bLambda_0}$ where $\bbeta_0 \in \Reals{p}, \bLambda_0 \in \Reals{d_0 \times d_0}, \bU_0 \in \mathcal{V}_{d_0, n}$, and $\norm{\bLambda_0}_0 = d_0$. Let $\bbP_0^{(n)}$ and $\bbE_0^{(n)}$ denote the probability and expectation under this true data generating process.  We consider the scenario that $d_0$ grows with $n$ while $p$ is fixed. As such, we must also increase the prior's truncation level $d$ as $n$ increases. We find that the following conditions on the growth rate of $d$ are sufficient to obtain our results:

\begin{cond}[Growth of $d$ with $n$]\label{assump:A3}
    $d = \lceil n^{\gamma} \rceil$ for some $\gamma \in (0, 1]$.
\end{cond}

\begin{cond}[Bounded scale parameter]\label{assump:A1}
    For sufficiently large $n$, $b \leq K d$ for a constant $K > 0$ independent of $n$.
\end{cond}

Condition \ref{assump:A3} requires $d$ to grow as a fractional power of $n$. Condition~\ref{assump:A1} constrains the scale of the Laplace slab. This condition holds when $b$ is a constant or grows as a fractional power (less than or equal to $\gamma$) of $n$, 
such as our choice of $b = \sqrt{n/2}$ in the following sections.

Next, we make the following assumptions on the truth:

\begin{assump}[Growth of $d_0$]\label{assump:A2}
    $d_0 / \log d \rightarrow  0$ as $n \rightarrow \infty$.
\end{assump}

\begin{assump}[Bounded $\bLambda_0$]\label{assump:A4}
$\sup_{1 \leq h \leq d_0}{\abs{\lambda_{0h}}} \leq K_{\lambda}$ for some $K_{\lambda} > 0$.
\end{assump}

\begin{assump}[Bounded latent space]
    $\max_{1\leq i \leq n} \norm{\bu_{0,i}}_2 \leq K_{u}$ for some $K_{u} > 0$.
\end{assump}

\begin{assump}[Bounded covariate effects]\label{assump:A6}
    $\norm{\bbeta_0}_2 \leq K_{\beta}$ for some $K_{\beta} > 0$. 
\end{assump}


\begin{assump}[Bounded covariates]\label{assump:A5}
    $\max_{1 \leq i \leq j \leq n} \norm{\bx_{ij}}_2 \leq K_x$ for some $K_x > 0$.
\end{assump}

Assumption \ref{assump:A2} sets the growth rate of the true latent space dimension $d_0$ relative to the truncation level $d$, which trivially holds for the common assumption that $d_0$ is bounded due to Condition \ref{assump:A3}. Assumptions \ref{assump:A4} -- \ref{assump:A5} are equivalent to bounding $\max_{1 \leq i \leq j \leq n}\abs{\bbeta^{\top}_0\bx_{ij} + \bu_{0i}^{\top}\bLambda_0\bu_{0j}}$,
which is common in the LSM literature~\citep{wu2017, ma2020}. 
For a Bernoulli GLNEM, these bounds imply that the resulting networks are dense. 

Finally, we need the following mild assumptions about the function $b(\cdot)$ and the link function $g(\cdot)$. These conditions are satisfied by all GLNEMs in this work.

\begin{assump}[Bounded variance]\label{assump:A7}
    For any compact subset $\mathcal{K} \subset \Theta$, there exists positive constants $K_{b,1}, K_{b,2} > 0$ such that $K_{b,1} \leq \inf_{\theta \in \mathcal{K}} b''(\theta) \leq \sup_{\theta \in \mathcal{K}} b''(\theta) \leq K_{b,2}$.
\end{assump}

\begin{assump}[Inverse link has a bounded derivative]\label{assump:A8}
    $\sup_{\set{\eta:\ \abs{\eta} \leq M}} (g^{-1})'(\eta) \leq K_{g}$ for some $K_g > 0$, where $M$ is the bound implied by Assumptions \ref{assump:A4} -- \ref{assump:A5}.
\end{assump}

The next theorem states that the average posterior probability that $\norm{\bLambda}_0$ overshoots a constant multiple of the true latent space dimension $d_0$ goes to zero as the number of nodes $n$ and prior truncation level $d$ simultaneously increase. 

\begin{theorem}[The Posterior Concentrates on Low Dimensions in Expectation]\label{thm:post_con}
    Assume model (\ref{eq:systematic}) -- (\ref{eq:edm}) with true non-zero latent space dimension $d_0$ and true parameters $\set{\bbeta_0, \bLambda_0, \bU_0}$ such that $\norm{\bLambda_0}_0 = d_0$. Assume the following prior: $\blambda \sim \text{SS-IBP}_d(1/d, d^{1 + \delta}, \bbP_{spike}, \bbP_{slab})$ for $\delta > 0$,  $\bbP_{spike} = \delta_0$, $\bbP_{slab} = \Laplace(b)$ for $b \geq 1$, $\bbeta \sim N(\mathbf{0}_p, \sigma_{\beta}^2 \bI_p)$, $\bU \in \mathcal{V}_{d,n}$ with prior probability one, and $b$ and $d$ satisfy \ref{assump:A3} - \ref{assump:A1}. Assume \ref{assump:A2} - \ref{assump:A8} hold, then 
\[
    \lim_{n \rightarrow \infty} \bbE_0^{(n)} \, \Pi\left(\, \norm{\bLambda}_0 > C d_0 \mid \bY^{(n)} \, \right) = 0,
\]
    for some $C > 1$ that only depends on $\delta$ and $K_{\lambda}$.
\end{theorem}

Practically, Theorem~\ref{thm:post_con} justifies the following strategy for estimating $d_0$. Without knowledge of $d_0$, one would set the prior truncation level $d$ as large as possible with the hope that the shrinkage prior will select an appropriate dimension. An objection to this strategy is that employing too many latent dimensions may cause the posterior to over-estimate the dimension of the latent space due to over-fitting. Theorem~\ref{thm:post_con} alleviates this concern by showing that the posterior will asymptotically concentrate on low-dimensional latent spaces near the true dimension $d_0$.

\section{Estimation}\label{ss:sec:estimation}

In this section, we develop a Markov chain Monte Carlo (MCMC) algorithm to sample from the GLNEM's posterior. Due to the high-dimensionality of the latent variables and the necessity to provide an algorithm that applies over a wide range of random and systematic components, we use gradient-based sampling methods. Specifically, we propose a Metropolis-within-Gibbs sampler with Hamiltonian Monte Carlo (HMC) proposals. However, the requirement that $\bU$ must lie in $\bar{\mathcal{V}}_{d,n}$ poses a challenge for gradient-based samplers because they often require the parameters to lie in an unconstrained space. As such, we first propose various parameter expansion strategies that overcome this limitation.


\subsection{A Prior Over $\bar{\mathcal{V}}_{d,n}$}\label{subsec:semi_ortho_prior}

An explicit prior density over centered semi-orthogonal matrices is difficult to construct due to the constraints placed on these matrices. In fact, we are unaware of any explicit prior density over $\bar{\mathcal{V}}_{d,n}$ currently used for Bayesian inference. Instead, we use the following probabilistic construction based on the QR decomposition:
\begin{enumerate}
\item Sample $\bB \in \Reals{n \times d}$ with entries independently drawn from a standard normal $N(0, 1)$.
\item Calculate the unique QR decomposition of $\tilde{\bB} = [\mathbf{1}_n, \bB] \in \Reals{n \times (d+1)}$, i.e., $\tilde{\bB} = \bQ \bR$, where $\bQ \in \mathcal{V}_{d+1,n}$ and $\bR \in \Reals{(d+1) \times (d+1)}$ is an upper-triangular matrix with positive elements on its diagonal.
\item Set $\bU$ to the last $d$ columns of $\bQ$, i.e., $\bU = \bQ_{2:(d+1)}$.
\end{enumerate}
This procedure generates a matrix $\bU \in \bar{\mathcal{V}}_{d,n}$ since the columns of $\bQ_{2:(d+1)}$ lie in the orthogonal complement of $\spn(\bQ_1) = \spn(\mathbf{1}_n)$, where $\bQ_1$ is the first column of $\bQ$. To make the dependence of $\bU$ on $\bB$ explicit, we denote $\bU_{\bB}$ as a matrix constructed in this fashion. Also, the following result holds.
\begin{proposition}\label{prop:prior_support}
The support of the probability distribution of $\bU_{\bB}$ equals $\bar{\mathcal{V}}_{d,n}$.
\end{proposition}
This proposition states that the probability distribution induced by steps 1--3 assigns positive probability to all open neighborhoods of $\bar{\mathcal{V}}_{d,n}$'s elements, which makes the distribution suitable as a prior when performing posterior inference over parameters in $\bar{\mathcal{V}}_{d,n}$.  Lastly, since this procedure is differentiable with respect to $\bB$~\citep{walter2012}, we can easily include it as a prior in gradient-based sampling methods such as the one we develop in Section~\ref{subsec:hmc_gibbs}.


\subsection{Hierarchical Representation of the SS-IBP with a Laplace Slab}\label{subsec:change_of_variable}

To improve the mixing of the MCMC algorithm, we use a hierarchical representation of the $\text{SS-IBP}_d(a, \kappa, \delta_0, \Laplace(b))$ prior based on the exponential scale mixture representation of the Laplace distribution~\citep{park2008}. The slab probabilities $\theta_1, \dots, \theta_d$ come from 
(\ref{eq:ssibp_prior}); however, we introduce binary indicator variables, $Z_1, \dots, Z_d$, such that for $h = 1, \dots, d$,
\begin{equation*}
    Z_h \mid \theta_h \indsim \Bern(\theta_h), \quad \sigma_h^2 \iidsim \Exp(1/2b^2), \quad \lambda_h \indsim Z_h N(0, \sigma_h^2) + (1 - Z_h) \delta_0.
\end{equation*}
Marginalizing over $\sigma_{1:d}^2$ and $Z_{1:d}$ recovers the original $\text{SS-IBP}_d(a, \kappa, \delta_0, \Laplace(b))$ prior. To further improve mixing, we use a non-centered parameterization~\citep{papas2007} for $\lambda_h$, that is, we set $\lambda_h = Z_h \sigma_h \tilde{\lambda}_h$, where $\tilde{\lambda}_h \iidsim N(0, 1)$ for $1 \leq h \leq d$. We denote $\bLambda(Z_{1:d}) = \diag(Z_1 \sigma_1 \tilde{\lambda}_1, \dots, Z_d \sigma_d \tilde{\lambda}_d)$ to make the dependence of $\bLambda$ on $Z_{1:d}$ clear.

\subsection{The Metropolis-within-Gibbs Sampler}\label{subsec:hmc_gibbs}

Here we outline the MCMC algorithm that we use to sample from the GLNEM's posterior. Let $\bpsi = \set{\bU_{\bB}, \bLambda(Z_{1:d}), \bbeta, \sigma_{1:d}^2, \theta_{1:d}}$ denote the collection of model parameters excluding the dimension indicator variables $Z_{1:d}$. Overall, the algorithm is a Gibbs sampler that alternates between sampling $\bpsi$ and $Z_{1:d}$ from their full conditional distribution $p(\bpsi \mid Z_{1:d}, \bY)$ and $p(Z_{1:d} \mid \bpsi, \bY)$, respectively. To sample from $p(\bpsi \mid Z_{1:d}, \bY)$, we use Hamiltonian Monte Carlo (HMC) with adaptive tuning parameter selection~\citep{neal2011, hoffman2014} as implemented in NumPyro~\citep{phan2019, bingham2019}. To sample from $p(\bpsi \mid Z_{1:d}, \bY)$ using NumPyro's HMC routine, we must provide the logarithm of this density up to an additive constant, which can be found in Section~\ref{sec:full_cond} of the supplementary material. The model parameters are initialized as detailed in Section~\ref{sec:mcmc_info} of the supplementary material. Algorithm~\ref{alg:hmc_within_gibbs} outlines the full Metropolis-within-Gibbs algorithm. 

Up until now, we have made the unrealistic assumption that the dispersion parameter $\phi$ is known. In addition, the Tweedie distribution has an  unknown power parameter $1 < \xi < 2$. Under the proposed Bayesian inference scheme, we can infer these parameters by giving them priors and analyzing their posteriors. We place a half-Cauchy prior on $\phi$ and a uniform prior on $\xi$, i.e., $\phi \sim C^{+}(0, 1)$ and $\xi \sim \textrm{Unif}[1.01, 1.99]$. The truncation of the endpoints for $\xi$ avoids the multimodality of the Tweedie density.  These parameters are then sampled as a part of step 1 in Algorithm~\ref{alg:hmc_within_gibbs}. Because the normalizing constant for the Tweedie distribution is intractable, we approximate it using the series expansion proposed by \citet{dunn2005}.

\begin{sampler}[t]
Given the previous parameters $\bpsi^{(s)}$ and $Z_{1:d}^{(s)}$, update the current parameters as follows:
\begin{enumerate}
\item Sample $\bpsi^{(s + 1)}$ from $p(\bpsi \mid Z_{1:d}^{(s)}, \bY)$ defined in Equation (\ref{eq:hmc_logpost}) of the supplementary material using HMC with adaptive tuning.
\item Let $\pi \in \texttt{Permutation}(\set{1, \dots, d})$ and set $Z_{1:d}^{(s+1)} = Z_{1:d}^{(s)}$. For $h = 1, \dots, d$:

Sample $Z_{\pi(h)}^{(s + 1)}$ from $p(Z_{\pi(h)} \mid Z_{-\pi(h)}^{(s+1)}, \bpsi^{(s+1)}, \bY)$, where
\[
\text{odds}(Z_{\pi(h)} = 1 \mid Z_{-\pi(h)}^{(s+1)}, \bpsi^{(s+1)}, \bY) = \frac{p(\bY \mid \bTheta(Z_{1:d}^{(s+1)}), Z_{\pi(h)} = 1, Z_{-\pi(h)}^{(s+1)})}{p(\bY \mid \bTheta(Z_{1:d}^{(s+1)}), Z_{\pi(h)} = 0, Z_{-\pi(h)}^{(s+1)})} \times \frac{\theta_{\pi(h)}^{(s+1)}}{1 - \theta_{\pi(h)}^{(s+1)}}.
\]
and
\[
    \bTheta(Z_{1:d}^{(s+1)}) =  (g \circ b')^{-1}\left(\sum_{k=1}^p \bX_k \beta^{(s+1)}_k  + \bU_{\bB}^{(s+1)} \bLambda^{(s+1)}(Z_{1:d}^{(s+1)}) \bU_{\bB}^{(s+1) \, \textrm{T}}\right),
\]
        with $(g \circ b')^{-1}$ applied element-wise. The likelihood terms are from model (\ref{eq:systematic}) -- (\ref{eq:linpred}).
\end{enumerate}
Note that $\texttt{Permutation}(\set{1, \dots, d})$ denotes a different random permutation of $\set{1, \dots, d}$ at each iteration. Also, we use $Z_{-h}$ to indicate the collection of $Z_{1:d}$ with the $h$-th element removed, i.e., $(Z_1, \dots, Z_{h-1}, Z_{h+1}, \dots, Z_d)$.
\caption{Metropolis-within-Gibbs Sampler for GLNEMs}
\label{alg:hmc_within_gibbs}
\end{sampler}

Although the $\SSIBP$ prior places soft identifiability constraints on the parameters, the posterior remains invariant to signed-permutations of $\bU$'s columns. As such, we post-process the posterior samples by matching to some reference configuration based on the Frobenious norm. We take the maximum a posteriori (MAP) estimates as the reference. This matching is a linear assignment problem, which we solve with the Hungarian method~\citep{kuhn1955}. In addition, $\bU$'s posterior mean is not guaranteed to lie in $\bar{\mathcal{V}}_{d,n}$. As such, we take a Bayesian decision theoretic approach and use $\bU$'s posterior Fr\'echet mean as a point estimate. See Section~\ref{sec:mcmc_info} of the supplementary material for details.

\section{Simulation Studies}\label{ss:sec:simulation}

In this section,  we present simulation studies that assessed different properties of the SS-IBP prior and our estimation method. We performed two studies that evaluated the following tasks: (1) the ability of the MCMC algorithm to recover the model parameters, and (2) the accuracy of the $\SSIBP$ prior in recovering the true dimension of the latent space compared to existing methods. An assessment of the method's sensitivity to model misspecification and zero-inflation is included in Section~\ref{sec:zero-inflation} of the supplement.

\subsection{Simulation Setup and Competing Methods}\label{ss:subsec:competition}

For various values of $n$, we generated 50 synthetic networks from the following five GLNEMs with parameters $\set{\bbeta_0, \bU_0, \bLambda_0}$ and true dimension $d_0$: (1) Bernoulli with a logistic link, (2) Gaussian with an identity link and variance $\sigma^2 = 9.0$, (3) Poisson with a log link, (4) negative binomial with a log link and dispersion $\phi = 0.5$, and (5) Tweedie with a log link, dispersion $\phi = 10$, and power parameter $\xi = 1.6$. For each network, we sampled an initial latent position matrix $\tilde{\bB} = (\tilde{\bb}_1, \dots, \tilde{\bb}_n)^{\top}$ by drawing $\tilde{\bb}_i \iidsim 0.5 N(\bmu, 0.1^2) + 0.5 N(-\bmu, 0.1^2)$, where $\bmu = (1/\sqrt{d_0})\mathbf{1}_{d_0}$. We set $\bU_0 = \bU_{\tilde{\bB}}$, where $\bU_{\tilde{\bB}}$ was constructed using the procedure described in Section~\ref{subsec:semi_ortho_prior}. We sampled $\bLambda_0$'s diagonal elements from a Gaussian mixture $\lambda_{0h} \iidsim 0.5 N(cn, n) + 0.5N(-cn, n)$ for $1 \leq h \leq d_0$. We set $c = 1$ for the Bernoulli and Gaussian models, $c = 0.5$ for the Poisson and negative binomial models, and $c = 2$ for the Tweedie model. We included five dyadic covariates. For each dyad, we set $x_{ij,1} = 1$ to account for an intercept and we drew $x_{ij,k} \iidsim \mathrm{Unif}[-1, 1]$ for $2 \leq k \leq 5$. We set $\bbeta_0 = (\beta_{01}, -0.5, 0.5, 0, 0)^{\top}$. We set the intercept $\beta_{01}$ to $-1$ for all models except the Gaussian model where we set it to $1$ to better match the positive values found in real data. We set $d_0 = 3$ in all simulations.

We compared our proposed SS-IBP prior with two alternative selection criteria commonly used in the latent space literature: minimization of information criteria and K-fold cross-validation. For the information criteria, we considered the Akaike information criterion (AIC), Bayesian information criterion (BIC), deviance information criterion (DIC), and the Watanabe-Akaike information criterion (WAIC). The K-fold cross-validation (K-fold CV) scheme used $K = 5$ folds, where the folds defined a non-overlapping partition of the network's dyads. We selected the dimension that maximized the held-out log-likelihood averaged across all folds. We found that the K-fold CV tended to overestimate the latent space dimension due to the held-out log-likelihood curves increasing drastically around the true dimension and slowly decreasing afterword. To compensate for this behavior, we included the one standard error rule, which selected the smallest model within one standard deviation of the best model (K-fold CV 1SE). The competing  methods placed a $N(0, n\mathbf{I}_d)$ prior on the diagonal of $\bLambda$ as in \citet{hoff2009}. We estimated these models by running NumPyro's HMC algorithm for 5,000 iterations after a burn-in of 5,000 iterations. Section~\ref{sec:model_select} of the supplementary material has more details on the competing methods. 

Lastly, we estimated the GLNEMs with SS-IBP priors by running the MCMC algorithm proposed in Section~\ref{ss:sec:estimation} for 5,000 iterations after a burn-in of 5,000 iterations. To match the theory developed in Section~\ref{sec:theory}, we set $a = 1 / d$ and $\kappa = d^{1.1}$. Furthermore, we used a truncation level of $d = 8$ in all simulations. For the remaining hyperparameters, we set $b = \sqrt{n/2}$ and $\sigma_{\beta} = 10$, which resulted in relatively broad priors.


\subsection{Parameter Recovery}\label{subsec:performance}

This section presents results that demonstrate the ability of the MCMC algorithm proposed in Section~\ref{ss:sec:estimation} to recover the model parameters as the number of nodes increases. To evaluate the estimators' accuracy, we compared the top three dimensions ranked by their posterior inclusion probabilities, $\hat{\bbP}(Z_h = 1 \mid \bY) = (1/S) \sum_{s=1}^S \mathbbm{1}\set{Z_h^{(s)} = 1}$, to the ground truth. Our evaluation metrics included the trace correlation between the ground truth and estimated latent positions, i.e., $\tr(\bU_0^{\top} \hat{\bU})/d_0$, and the relative error with respect to the squared Frobenius norm between the true and estimated $\bLambda$ matrix and coefficients $\bbeta$. We also include the relative error between the true and estimated matrix $\bU \bLambda \bU^{\top}$, which was estimated using all eight latent dimensions.

Table~\ref{tab:performance} displays the results. Note that the trace correlation ranges from $-1$ to 1, with one indicating a perfect alignment of the subspaces spanned by $\bU_0$ and $\hat{\bU}$. For the relative errors, smaller is better.  In all cases, the MCMC algorithm is able to recover the model's parameters with high accuracy. Furthermore, the performance increases with the number of nodes. As such, we can conclude that the MCMC algorithm is successful in estimating the model parameters for a wide variety of GLNEMs.


\begin{table}[h]
\centering\small
\begin{tabular}{@{}lllllll@{}} \toprule
    GLNEM & $n$ & $\bU$ Trace Corr. & $\bLambda$ Rel. Error & $\bU \bLambda \bU^{\top}$ Rel. Error & $\bbeta$ Rel. Error \\
\bottomrule
 & 100 & 0.909 (0.074) & 0.004043 (0.003925) & 0.156 (0.021) & 0.01157 (0.0084) \\
Bernoulli & 200 & 0.965 (0.066) & 0.000909 (0.000826) & 0.078 (0.007) & 0.00313 (0.0018) \\
 & 300 & 0.977 (0.057) & 0.000380 (0.000337) & 0.052 (0.004) & 0.00102 (0.0011) \\[0.25em]
\hline
 & 100 & 0.926 (0.067) & 0.002575 (0.002996) & 0.173 (0.027) & 0.01325 (0.0086) \\
Gaussian & 200 & 0.960 (0.052) & 0.000581 (0.000524) & 0.086 (0.009) & 0.00389 (0.0025) \\
 & 300 & 0.975 (0.066) & 0.000179 (0.000195) & 0.058 (0.005) & 0.00138 (0.0012) \\[0.25em]
\hline
 & 100 & 0.945 (0.052) & 0.001652 (0.002492) & 0.141 (0.051) & 0.00250 (0.0024) \\
Poisson & 200 & 0.977 (0.040) & 0.000407 (0.000291) & 0.066 (0.015) & 0.00069 (0.0006) \\
 & 300 & 0.986 (0.141) & 0.000147 (0.313174) & 0.044 (0.167) & 0.00029 (0.0099) \\[0.25em]
\hline
 & 100 & 0.924 (0.172) & 0.002722 (0.195986) & 0.199 (0.173) & 0.00521 (0.1216) \\
Neg. Bin. & 150 & 0.951 (0.033) & 0.001147 (0.001195) & 0.132 (0.034) & 0.00217 (0.0017) \\
 & 200 & 0.965 (0.041) & 0.000599 (0.000543) & 0.100 (0.018) & 0.00139 (0.0008) \\[0.25em]
\hline
 & 50 & 0.935 (0.079) & 0.009565 (0.008941) & 0.118 (0.026) & 0.17972 (0.0983) \\
Tweedie & 100 & 0.969 (0.055) & 0.001007 (0.000831) & 0.043 (0.006) & 0.02304 (0.0157) \\
 & 150 & 0.982 (0.057) & 0.000313 (0.000257) & 0.027 (0.003) & 0.00990 (0.0060) \\
\bottomrule
\end{tabular}
\caption{Trace correlation and relative estimation error of the model's parameters as the number of nodes ($n$) in the networks increased. Each cell contains the median value of the metric with the standard deviation over the 50 simulation displayed in parentheses.}
\label{tab:performance}
\end{table}


\subsection{Dimension Selection}\label{subsec:sim_dim_select}

Here, we assess the ability of the proposed spike-and-slab prior to select the correct number of dimensions. We used the same simulation setup described in Section~\ref{ss:subsec:competition} with true latent space dimension $d_0 = 3$. Figure~\ref{fig:heatmap_canonical} and Figure~\ref{fig:heatmap_noncanonical} display heatmaps of the percentage of times a dimension was selected by the proposed SS-IBP prior and its competitors over the 50 simulations for GLNEMs with canonical and non-canonical link functions, respectively. Note that the proposed SS-IBP estimates the latent space dimension based on the posterior mode, that is, $\argmax_{1 \leq k \leq d} \hat{\mathbb{P}}(\sum_{h=1}^d  Z_h = k \mid \bY)$. The proposed SS-IBP prior performed the best or equivalent to the best in all scenarios. In contrast, the DIC and WAIC estimators often overestimate the number of latent space dimensions. The BIC estimator often underestimated the latent space dimension except for the Bernoulli case where it has competitive performance for larger networks. The K-fold CV estimator performed poorly except for the negative binomial distribution. The ad-hoc K-fold CV 1SE performed well for the Bernoulli and Gaussian models, but poorly otherwise. Note that we did not compute the cross-validation estimators for the largest network sizes of the negative binomial and Tweedie models because they took more than a day to compute for a single simulation. Finally, the AIC estimator performed well for the Bernoulli, Gaussian, and negative binomial models, but performed poorly for the Poisson and Tweedie models. Overall, unlike the competing methods, the SS-IBP prior has excellent performance for a variety of GLNEMs. The SS-IBP prior is also less computationally intensive than the competitors because it requires only a single MCMC run. The SS-IBP prior also inherently accounts for the uncertainty due to the latent space dimension, unlike the competitors.

\begin{figure}[tp]
    \centering
    \begin{subfigure}[b]{\textwidth}
        \centering
        \includegraphics[width=\textwidth]{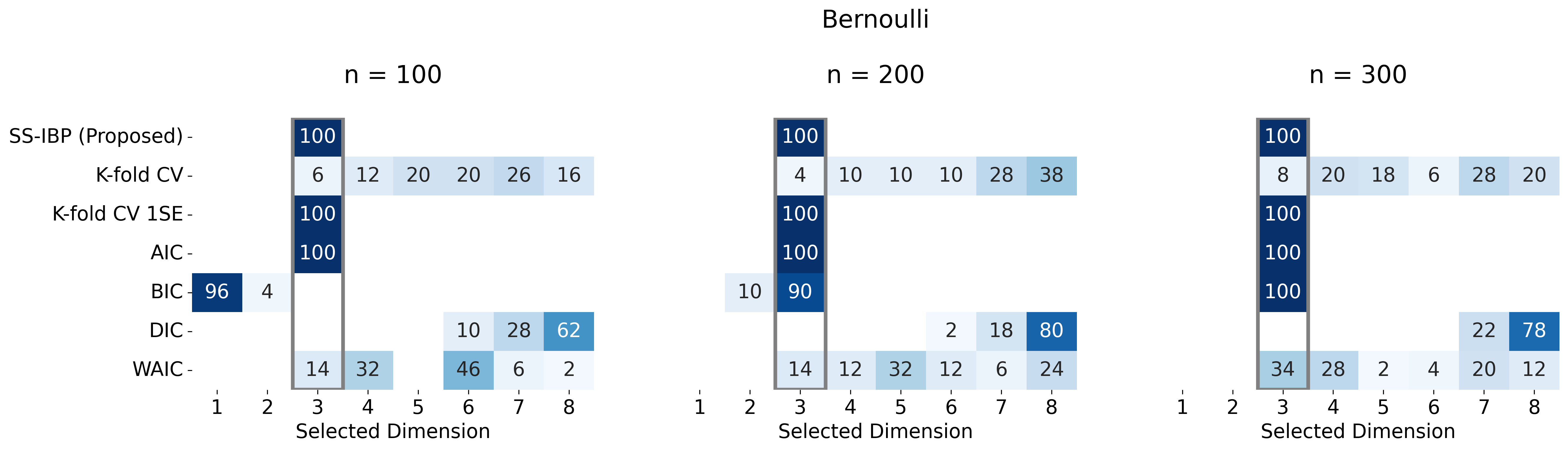}
    \end{subfigure}
    \begin{subfigure}[b]{\textwidth}
        \centering
        \includegraphics[width=\textwidth]{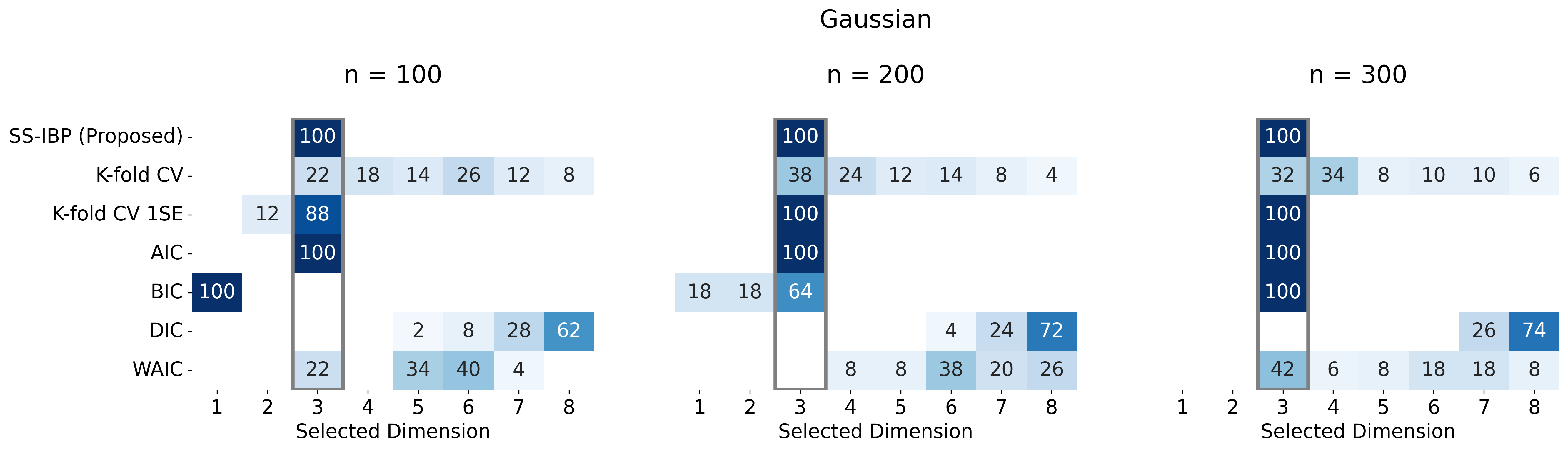}
    \end{subfigure}
    \begin{subfigure}[b]{\textwidth}
        \centering
        \includegraphics[width=\textwidth]{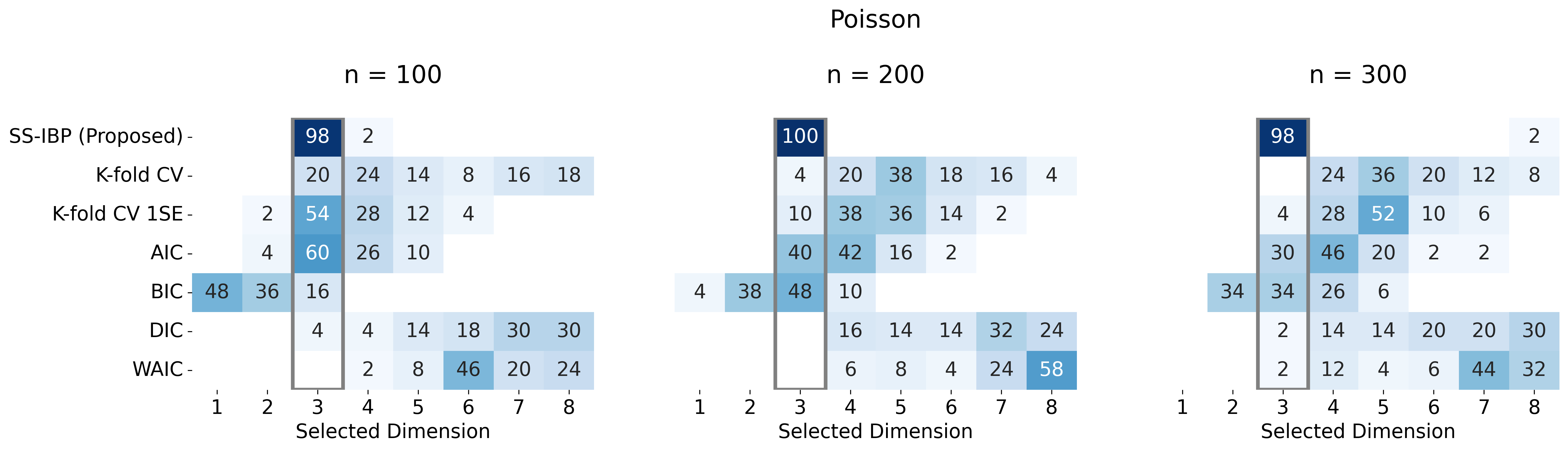}
    \end{subfigure}
    \caption{Heatmaps displaying the percentage of the 50 simulations in which the selected dimension was a particular value for the SS-IBP (proposed), K-fold CV, K-fold CV (1SE), and information criterion based estimators for GLNEMs with a canonical link function. The column corresponding to the true value of $d_0 = 3$ is outlined in gray. Darker blue cells indicate percentages closer to $100\%$.} 
    \label{fig:heatmap_canonical}
\end{figure}

\begin{figure}[t]
    \centering
    \begin{subfigure}[b]{\textwidth}
        \centering
        \includegraphics[width=\textwidth]{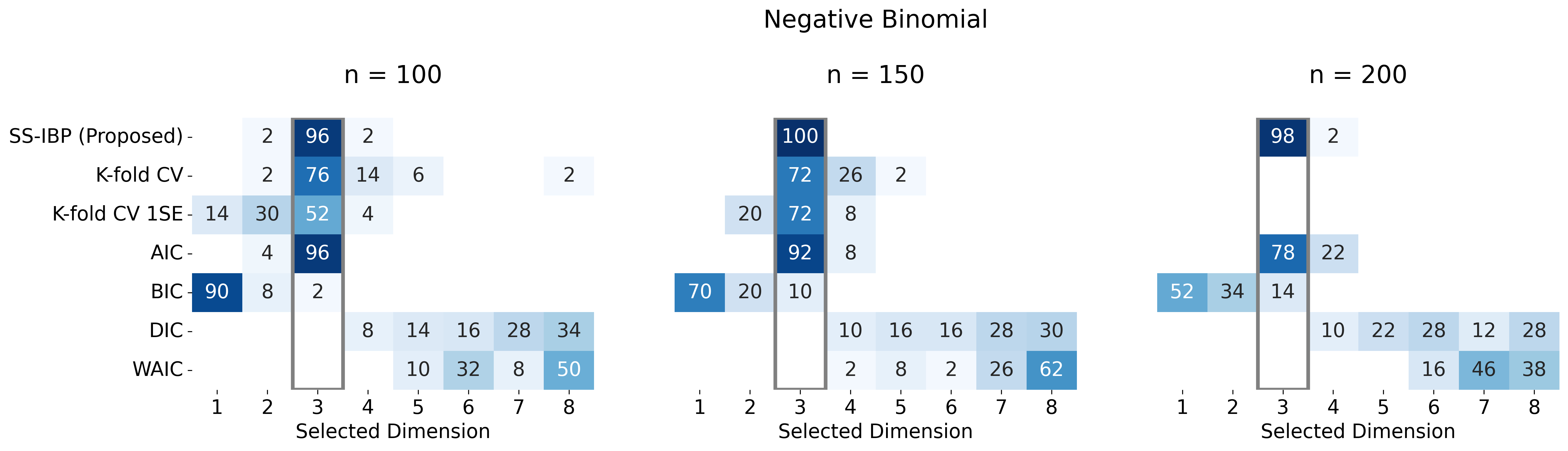}
    \end{subfigure}
    \begin{subfigure}[b]{\textwidth}
        \centering
        \includegraphics[width=\textwidth]{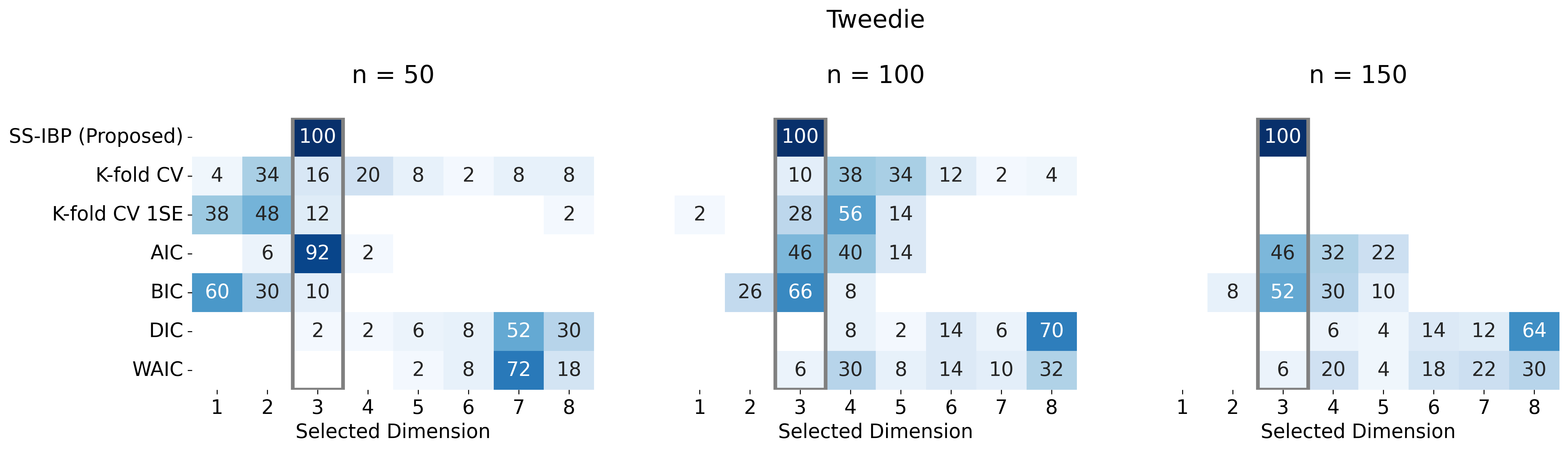}
    \end{subfigure}
    \caption{Heatmaps displaying the percentage of the 50 simulations in which the selected dimension was a particular value for the SS-IBP (proposed), K-fold CV, K-fold CV (1SE), and information criterion based estimators for GLNEMs with a non-canonical link function. The column corresponding to the true value of $d_0 = 3$ is outlined in gray. Darker blue cells indicate percentages closer to $100\%$. The CV estimators were not calculated for the largest network sizes because of their prohibitively long computation time.}
    \label{fig:heatmap_noncanonical}
\end{figure}


\section{Real Data Applications}\label{sec:applications}

To illustrate the benefits of the proposed methodology, we used it to analyze networks arising from different scientific fields, namely biology, ecology, and economics. The first example studies a binary network of interactions between 270 proteins of \textit{E.\ coli}. The second example investigates the similarities between 51 tree species based on a count network of shared fungal parasites. Finally, the third example analyzes trade in bananas between 75 nations based on a network with non-negative continuous edges. All GLNEMs were estimated using the same hyperparameter values outlined in the simulation studies.


\subsection{Protein-Protein Interactions}\label{subsec:ppnetwork}

Here, we re-analyze a binary network of protein-protein interaction data from \citet{butland2005} that was introduced as a benchmark for the eigenmodel by \citet{hoff2009}. Each node corresponds to an essential protein of \textit{E.\ coli} with $n = 270$ proteins overall. A binary edge ($Y_{ij} = 1$) is recorded if protein $i$ and protein $j$ interact, otherwise $Y_{ij} = 0$. 

In this analysis, we demonstrate that the proposed SS-IBP prior selects a better fitting model than the K-fold cross-validation estimator used in previous analyses~\citep{hoff2008, hoff2009, jauch2021}. Specifically, we estimated a Bernoulli GLNEM with a logistic link using both our proposed SS-IBP prior and the $N(\mathbf{0}_d, n \bI_d)$ prior used in the original analysis. Note that the original analysis used a probit link; however, we use a common logistic link to focus on the differences due to model size. Since there are no dyadic covariates, the model only contains an intercept. We estimated both models by running the MCMC algorithms for 15,000 iterations after a burn-in of 7,500 iterations. The Gaussian model using the K-fold CV 1SE estimator with $K = 4$ and $1 \leq d \leq 8$ estimated $\hat{d}_0 = 3$, which matches the original analysis.  For the SS-IBP prior we set $d = 10$ and made the same hyperparameter choices as described in Section~\ref{ss:subsec:competition}. The posterior mode selected $\hat{d}_0 = 5$ dimensions. See Figure~\ref{fig:ppnetwork_kfold} and Figure~\ref{fig:pp_network_mcmc} in Section~\ref{ss:subsec:additional_figures} of the supplement for summaries of these selection procedures.

A natural approach to comparing the goodness-of-fit of two network models is based on their predictions of observed quantities~\citep{hunter2008}. As such, we demonstrate that the proposed model better describes the data by comparing the two models' posterior predictive distributions of the transitivity coefficient and the degree distribution. We chose to use the transitivity coefficient because latent space models are often used to explain the high levels of transitivity found in real-world networks. We included the degree distribution because the network science literature puts tremendous focus on this statistic. Figure~\ref{fig:ppnetwork_gof} displays the results. The proposed model better fits the data based on both statistics. The three-dimensional model underestimates the transitivity, while the observed transitivity falls within the SS-IBP model's posterior predictive distribution. Similarly, the degree distribution more closely matches the SS-IBP model, although both models underestimate the number of single degree nodes and overestimate the number of isolated nodes.

\begin{figure}[htb]
\centering \includegraphics[width=\textwidth, keepaspectratio]{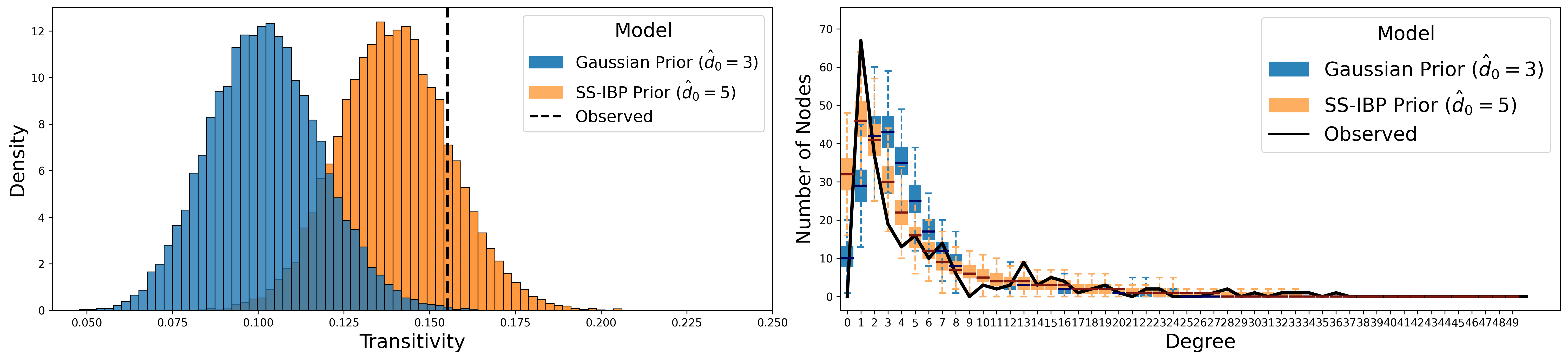}
    \caption{Goodness-of-fit statistics comparing a Bernoulli GLNEM with logistic link and $\blambda$ given a $N(\mathbf{0}_{d_0}, n \mathbf{I}_{d})$ prior with $\hat{d}_0 = 3$ selected by K-fold CV 1SE with $K = 4$ and the proposed SS-IBP prior ($\hat{d}_0 = 5$) on a binary network of protein-protein interactions.}
\label{fig:ppnetwork_gof}
\end{figure}

\subsection{Host-Parasite Interactions Between Tree Species}\label{subsec:tree}

In ecology, networks are often used to study the extent that species are related. Here, we consider an ordinal-valued network of $n = 51$ tree species introduced by \citet{vacher2008}, where the edge variables $Y_{ij}$ represent the number of shared fungal parasites between tree species $i$ and $j$.  In addition, three dyadic covariates are measured that represent different distances between pairs of tree species: taxonomic ($\texttt{Tax}_{ij}$), geographic ($\texttt{Geo}_{ij}$), and genetic ($\texttt{Gene}_{ij}$) distances. This network was previously analyzed using a Poisson stochastic block model with covariates in \citet{mariadassou2010} and \citet{donnet2021}. 

Similar to \citet{mariadassou2010}, our goals are to understand how the covariates affect the number of shared fungal parasites and to describe any potential residual network structure after accounting for the covariates. We consider two models for networks with ordinal edge variables: A Poisson and negative binomial GLNEM both with a log link and an SS-IBP prior. We included the negative binomial model because the network is potentially zero-inflated with 46\% of the edges equaling zero. We estimated the models by running the MCMC algorithm for 15,000 iterations after a burn-in of 7,500 iterations.

To choose between the models, we examined the dispersion parameter of the negative binomial GLNEM, which reduces to the Poisson GLNEM when $\phi$ is zero. The 95\% credible interval for $\phi$ is (0.43, 0.62), which is well above zero and provides support for the negative binomial GLNEM. Furthermore, the WAIC of the Poisson and negative binomial models  are 3957 and 3807, respectively, which also favors the negative binomial model. The negative binomial model also has a lower-dimensional latent space with 
the one-dimensional space receiving 95\% of the posterior probability. On the other hand, the Poisson model was split between a latent space with two (posterior probability of $45\%$) and three (posterior probability of $48\%$) dimensions. See Figure~\ref{fig:tree_poisson_mcmc} and Figure~\ref{fig:tree_negbinom_mcmc} in Section~\ref{ss:subsec:additional_figures} of the supplement for posterior summaries. As demonstrated in the simulation study in Section~\ref{sec:zero-inflation} of the supplement, the negative binomial model's lower dimension may be due to zero-inflation.

\begin{figure}[htb]
\centering \includegraphics[width=\textwidth, keepaspectratio]{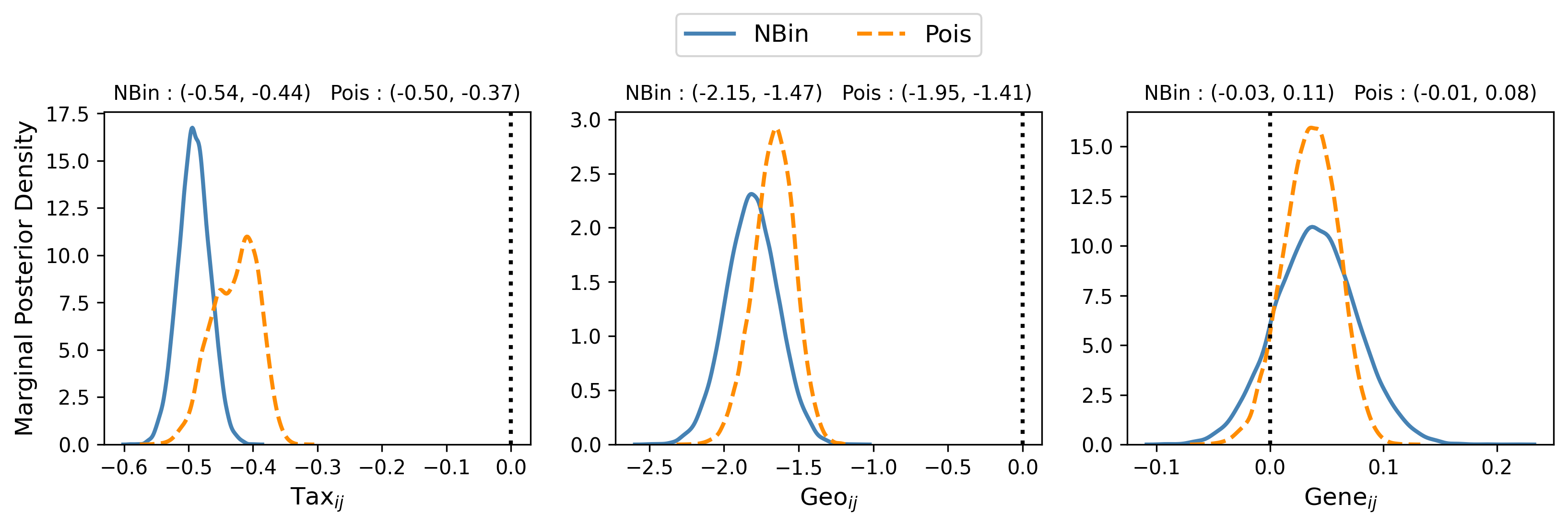}
    \caption{Marginal posterior densities of the coefficients for the three dyad-variables ($\texttt{Tax}_{ij}$, $\texttt{Geo}_{ij}$, $\texttt{Gene}_{ij}$) for the negative binomial (NBin) and Poisson (Pois) GLNEMs. The 95\% credible intervals are displayed on the top and the vertical dotted line is at zero.}
\label{fig:tree_covariates}
\end{figure}

Despite the two models' differences, Figure~\ref{fig:tree_covariates} demonstrates that their inferences about the covariates' marginal effects are roughly the same. Both models indicate that increasing a tree's average taxonomic or geographic distance to other trees in the network decreases the geometric mean of its expected number of shared parasites, while genetic distance is insignificant when these variables are in the model. These conclusions match the original analysis by \citet{mariadassou2010}; however, their results are conditional on the network structure. Also, note that under the Poisson GLNEM, the posterior density of $\texttt{Tax}_{ij}$'s coefficient has two local modes due to the mixing of the two and three dimensional models.


\begin{figure}[tb]
\centering \includegraphics[height=0.225\textheight, width=\textwidth,keepaspectratio]{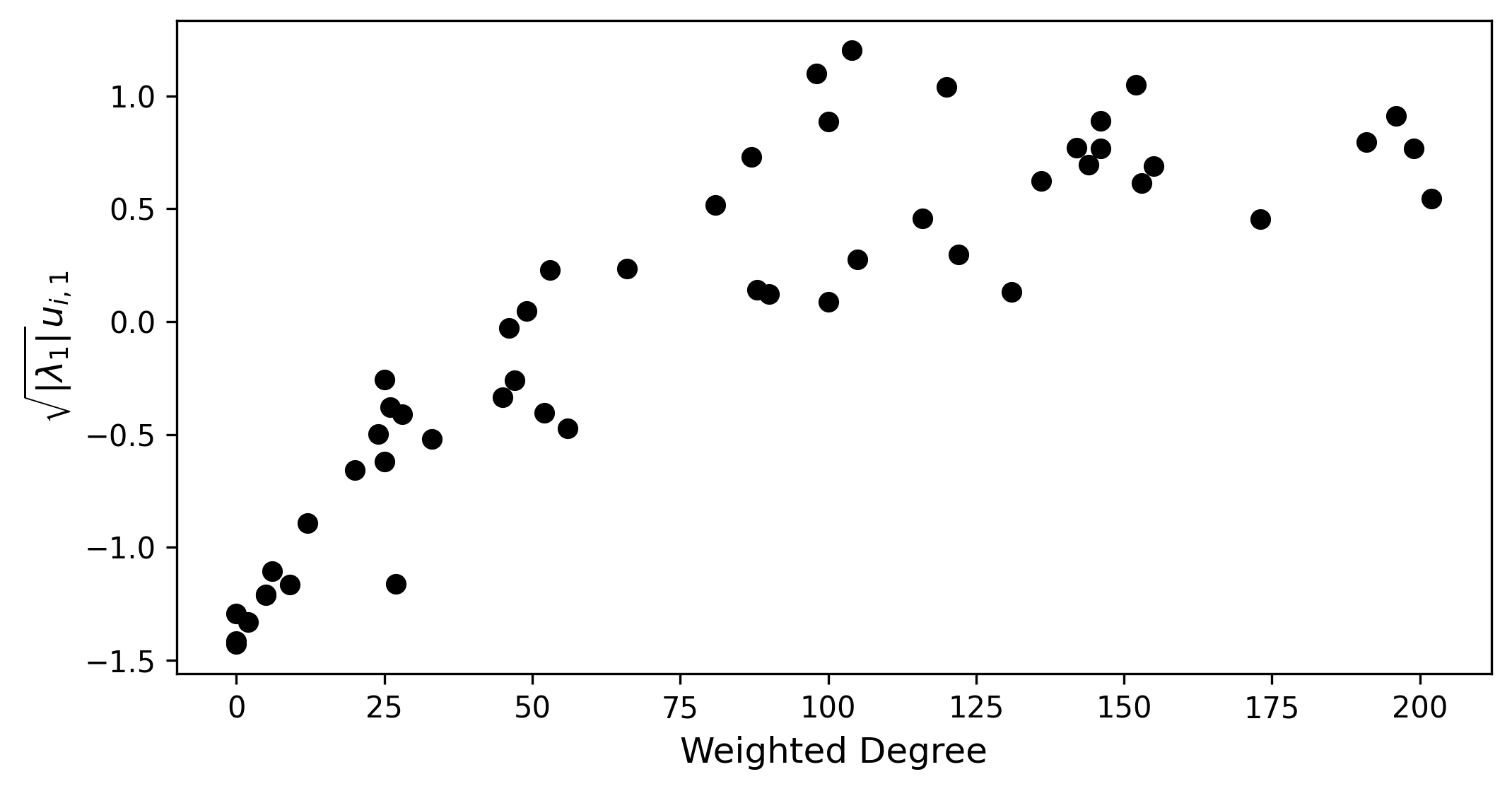}
    \caption{Scatter plot of the latent positions from the negative binomial GLNEM as a function of the weighted degree, $\sum_{j=1}^n y_{ij}$, for the tree network.}
\label{fig:tree_ls}
\end{figure}

Next, we analyze the residual network structure captured by the one-dimensional latent space in the negative binomial model. The scatter plot in Figure~\ref{fig:tree_ls} demonstrates that a node's latent position is highly correlated with its weighted degree with a Pearson correlation equal to 0.86. The 95\% credible interval for $\lambda_1$ is (26.61, 34.42) indicating assortativity along this dimension. These observations suggest that the latent space acts as a degree correction in the model. In conclusion, the negative binomial GLNEM fits the network data better than the Poisson GLNEM and also has a simpler latent structure.

\subsection{International Trade of Bananas in 2018}\label{subsec:banana}

Lastly, we consider the international trade of bananas, which is a topic of academic research~\citep{josling2003} and listed as one of the Food and Agriculture Organization's most important commodities for both consumption and trade~\citep{benedictis2014}. The data take the form of a non-negative continuous-valued network of nations where the edge variable $Y_{ij}$ is the amount of trade of bananas in thousands of U.S. Dollars between nation $i$ and nation $j$ in 2018. We only consider the top $n = 75$ nations in terms of total trade volume, that is, their weighted degree.  The network was constructed from the BACI database~\citep{cepii2010}, which is derived from the UN Comtrade database by the Centre d'\'{E}tudes Prospectives et d'Informations Internationales (CEPII). 

Our goal is to quantify the effect of various dyadic covariates on the trade of bananas, while controlling for network effects. To develop an appropriate GLNEM, we start with a model for trade from economics known as the gravity model~\citep{tinbergen1962, anderson1979}. In its simplest form, the gravity model states that the systematic component is
\begin{equation}\label{eq:gravity}
    \log(\mathbb{E}[Y_{ij} \mid \bx_{ij}]) = \beta_1 [\log(\texttt{GDP}_i) + \log(\texttt{GDP}_j)] + \beta_2 \log(\texttt{Dist}_{ij}) + \sum_{k=3}^p \beta_k \, x_{ij,k}, 
\end{equation}
where $\texttt{GDP}_i$ is 
nation $i$'s GDP in thousands of U.S. Dollars, $\texttt{Dist}_{ij}$ is the population weighted harmonic distance between nations $i$ and $j$, and the $x_{ij,k}$'s are any other relevant covariates. We included three additional binary indicator variables for shared language ($\texttt{CommonLang}_{ij}$), shared border ($\texttt{Border}_{ij}$), and active trade agreement ($\texttt{TradeAgreement}_{ij}$). We obtained the dyadic covariates from the CEPII Gravity database~\citep{gravity2022}.

For the random component, we propose to use the Tweedie distribution supported on the non-negative reals with power parameter $1 < \xi < 2$, which has not appeared in the network latent space literature. This distribution has been recognized as a suitable distribution for trade data~\citep{barabesi2016} because of its interpretation as a compound Poisson-gamma distribution and invariance to the edge variable's unit of measurement. See Section~\ref{sec:tweedie} of the supplementary material for more details. We estimated a Tweedie GLNEM with a log link and $d = 10$ by running the MCMC algorithm for 15,000 iterations after a burn-in of 7,500 iterations. The fraction of deviance explained~\citep{hastie2015} of the model is 0.75, the AUC (area under the reciever operating characteristic curve) for classifying zero-valued edges is 0.853, and the fraction of variance explained $R^2$ for non-zero edges is 0.86. All values are close to one, which indicates a good fit to the trade network.

Figure~\ref{fig:trade_covariates} displays the posterior distributions of the covariate effects and their 95\% credible intervals. According to the credible intervals, all covariates effects are significant. Furthermore, the sign of the coefficients for distance and GDP agree with those posited by economic theory, namely, that the amount of trade between nations decreases with distance and increases with GDP. The other coefficients indicate that sharing a language, a border, or entering into a trade agreement increases the volume of trade in bananas. 

\begin{figure}[htb]
\centering \includegraphics[height=0.4\textheight, width=\textwidth, keepaspectratio]{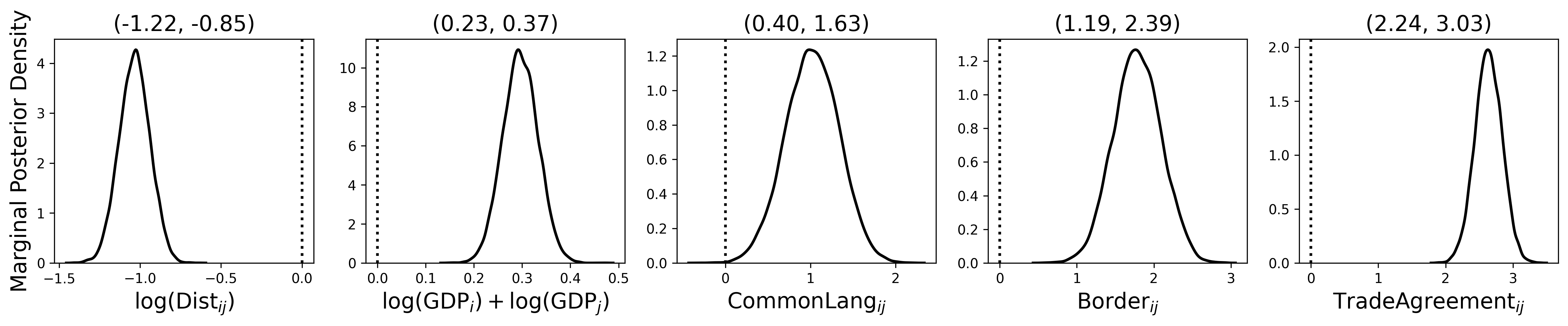}
\caption{Marginal posterior densities of the coefficients for the five dyad-variables for the Tweedie GLNEM on the banana trade network. The 95\% credible intervals are displayed on the top and the vertical dotted line is at zero.}
\label{fig:trade_covariates}
\end{figure}

\begin{figure}[tb]
\centering \includegraphics[height=0.39\textheight, width=\textwidth, keepaspectratio]{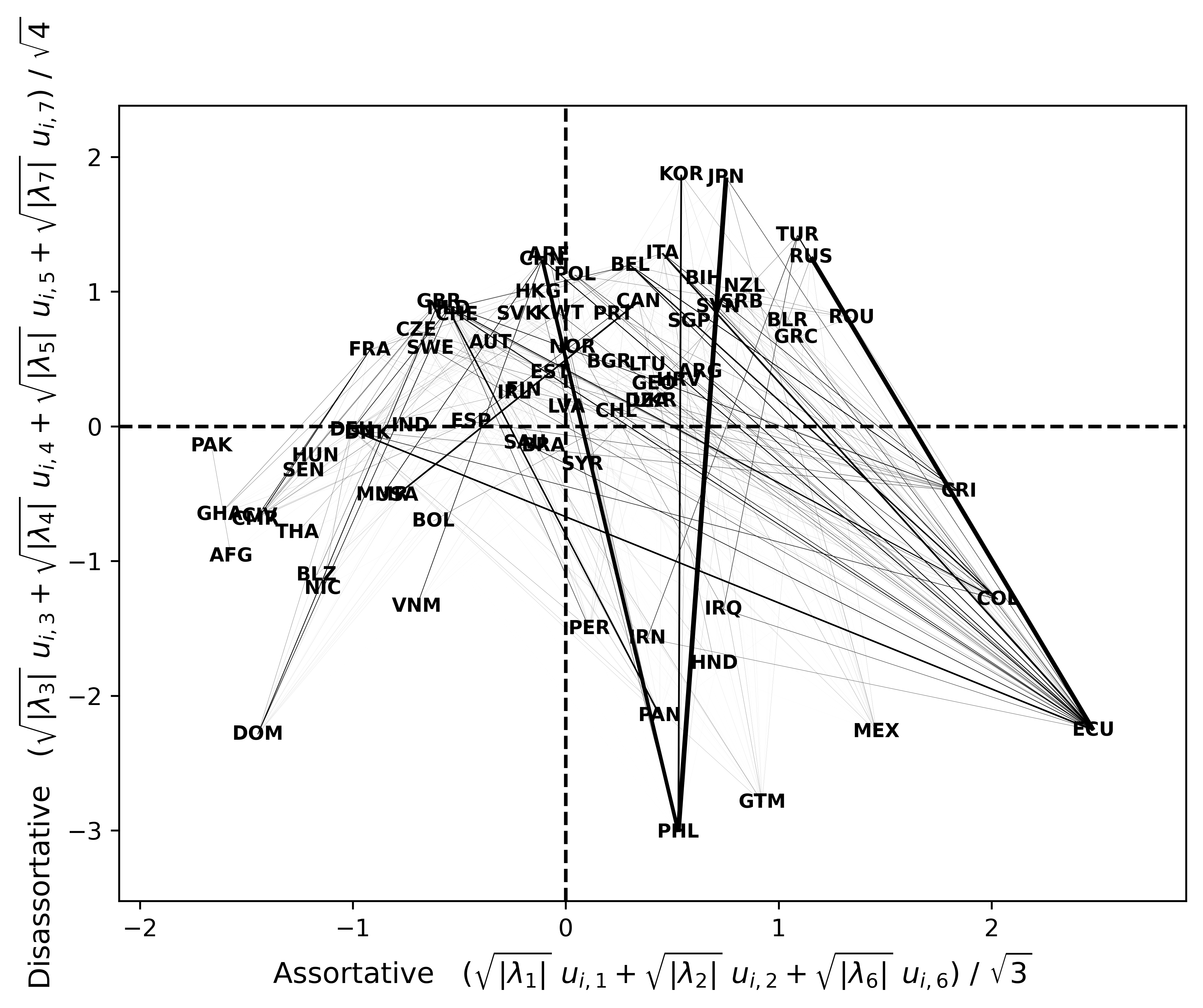}
    \caption{Two-dimensional summary of the latent space for the banana trade network. The $x$-axis (assortative summary) and $y$-axis (disassortative summary) are denoted by the dashed horizontal and vertical lines, respectively. Each nation's name is annotated and the observed edges are displayed with a width proportional to the trade amount.}
\label{fig:trade_ls}
\end{figure}

Next, we analyzed the residual network structure captured by the latent space. The dimension of the latent space was uncertain with the posterior consisting of a mixture of five (posterior probability 39\%), six (posterior probability 40\%), and seven (posterior probability 17\%) dimensional latent spaces. Detailed summaries are in Section~\ref{ss:subsec:additional_figures} of the supplement. Figure~\ref{fig:trade_ls} shows a two-dimensional visualization of the seven-dimensional latent space obtained by separately combining the assortative and disassorative dimensions into single-feature summaries, that is, $(\sqrt{\abs{\lambda_3}} u_{i,3} + \sqrt{\abs{\lambda_4}} u_{i,4} + \sqrt{\abs{\lambda_5}} u_{i,5} + \sqrt{\abs{\lambda_7}} u_{i,7})/\sqrt{4}$ and $(\sqrt{\abs{\lambda_1}} u_{i,1} + \sqrt{\abs{\lambda_2}} u_{i,2} + \sqrt{\abs{\lambda_6}} u_{i,6})/\sqrt{3}$. In this visualization, a nation tends to trade more with nations that they are closer to along the $x$-axis (assortative) and further apart along the $y$-axis (disassortative). The disassortative summary identifies a group of nations below the $x$-axis that primary trade with nations above the $x$-axis and do not trade with each other. These nations include the major global exporters of bananas such as Ecuador (ECU), the Philippines (PHL), and the Dominican Republic (DOM). The assortative summary encodes regional patterns in the trade relationships. For example, the Philippines primarily trade with nations in Asia such as Japan (JPN), South Korea (KOR), and China (CHN). In addition, the Dominican Republic focuses its trade with nations in the European Union such as the United Kingdom (GBR), the Netherlands (NLD), and Germany (DEU).

\section{Discussion}\label{sec:discussion}

In this article, we developed a theoretically supported Bayesian approach to dimension selection for a class of network models we called generalized linear network eigenmodels. 
This model class is very general and applicable to networks with edge variables drawn from any exponential-dispersion family and a systematic component using possibly non-canonical link functions. Furthermore, we introduced a new sum-to-zero identifiability constraint on the latent positions that allows one to interpret the effect of dyadic covariates as marginal effects. To enforce this constraint in a fully Bayesian way, we proposed a new prior for the latent positions based on the QR decomposition that has full support and easily integrates with gradient based sampling methods.

Next, we introduced a novel $\SSIBP$ prior that addresses the important problem of automatically selecting the latent space dimension and assessing its uncertainty. We showed that this prior induces a stochastic ordering that counteracts certain identifiability issues in the model's likelihood. In addition, we provided a posterior tail bound on the number of dimensions that reflects the true model's underlying dimensionality. To the best of our knowledge, this is the first such theoretical result for Bayesian LSMs. Furthermore, this result is a first step toward establishing a posterior concentration rate for GLNEM estimation. Note that our theoretical results do not establish posterior dimension selection consistency, which would require a matching lower bound on the number of dimensions. Lastly, we developed a general Metropolis-within-Gibbs MCMC algorithm applicable to any GLNEM with SS-IBP priors, and showed the flexibility of our method by modeling three real-world networks from biology, ecology, and economics.

An interesting research direction is to extend our approach to directed networks. The typical eigenmodel for directed networks involves a decomposition based on the singular value decomposition: the latent space component, $\bU \bLambda \bU^{\top}$,  becomes $\bU \bS \bV^{\top}$, where $\bU \in \mathcal{V}_{d,n}$, $\bV \in \mathcal{V}_{d,n}$, and the matrix of singular values $\bS = \diag(s_1, \dots, s_d)$ has positive real-valued entries. The positivity of the singular values invalidates our use of a $\Laplace(b)$ slab, so choosing an appropriate prior requires additional computational and theoretical care. Finally, while our approach to dimension selection is more computationally efficient than existing model selection criteria for Bayesian LSMs, our use of MCMC sampling is still computationally intensive for large networks. The exploration of an appropriate EM algorithm~\citep{dempster1977} to estimate the posterior mode or a variational inference algorithm~\citep{wainwright2008} to approximate the posterior is an area of future research interest.

\bibliographystyle{asa}
\bibliography{references}

\clearpage

\begin{appendix}

\begin{center}
{\Large\bf Supplementary Material for \\
    ``A Spike-and-Slab Prior for Dimension Selection in Generalized Linear Network Eigenmodels''}\\[1em]

{\large Joshua Daniel Loyal and Yuguo Chen}
\end{center}

\section{Proofs of Main Results}\label{sec:proofs}

    This section contains proofs of the various results found in the main text. Section~\ref{sec:proofs_prop} contains proofs of the various propositions concerning the identifiability of GLNEMs and the proposed SS-IBP prior. Section~\ref{sec:post_con} contains proofs of the tail-bound of the SS-IBP prior and the GLNEM's dimension under the posterior. Lastly, Section~\ref{sec:aux_lemma} contains auxiliary lemmas necessary to prove the results in the previous two sections.

    We use the following notation throughout this section. For two densities $p$ and $q$, we let $K(p, q) = \mathbb{E}_p \log(p/q)$ and $V(p, q) = \mathbb{E}_p [\log(p/q) - K(p, q)]^2$ denote the Kullback-Leibler (KL) divergence and second moment of the KL ball, respectively. For two positive sequences $\set{a_n}$ and $\set{b_n}$, we write $a_n \asymp b_n$ to denote $0 < \liminf_{n\rightarrow \infty} a_n / b_n \leq \limsup_{n \rightarrow \infty} a_n / b_n < \infty$. If $\lim_{n\rightarrow \infty} a_n / b_n = 0$, we write $a_n = o(b_n)$. We use $a_n = O(b_n)$ to denote that for sufficiently large $n$, there exists a constant $K > 0$ independent of $n$ such that $a_n \leq K b_n$. For a vector $\bv = (v_1, \dots, v_p)^{\top} \in \Reals{p}$, we let $\norm{\bv}_0, \norm{\bv}_2, \norm{\bv}_{\infty}$ denote the $\ell_0$, $\ell_2$, and $\ell_{\infty}$ norm, respectively. For a matrix $\bC \in \Reals{m \times m}$, we denote the Frobenius norm as $\norm{\bC}_F = (\sum_{ij} C_{ij}^2)^{1/2}$. Last, for two real numbers $a, b \in \Reals{}$, $a \vee b = \max(a, b)$ and $a \wedge b = \min(a, b)$.

\subsection{Proofs of Propositions \ref{prop:identify}, \ref{prop:ssibp-ordering}, and \ref{prop:prior_support}}\label{sec:proofs_prop}

\begin{proof}[Proof of Proposition~\ref{prop:identify}]

    The condition that two probability distributions under two different parameterizations coincide implies that their moments must also coincide. This implies that the adjacency matrix's expected value $\bbE[\bY \mid \bX_{1:p}]$ is identifiable. Furthermore, the condition that the link function $g$ is strictly increasing, thus invertible, implies that 
\begin{equation}\label{eq:identify}
    \sum_{k=1}^p \bX_k \beta_k + \bU \bLambda \bU^{\top} = \sum_{k=1}^p \bX_k \tilde{\beta}_k + \tilde{\bU} \tilde{\bLambda} \tilde{\bU}^{\top}.
\end{equation}
Now, note that
\begin{equation}\label{eq:one_mul}
    \left(\sum_{k=1}^p \bX_k \beta_k\right) \mathbf{1}_n = 
    \begin{pmatrix}
        \sum_{k=1}^p \beta_k \sum_{j=1}^n x_{1j,k}  \\
        \vdots \\
        \sum_{k=1}^p \beta_k \sum_{j=1}^n x_{nj,k}
    \end{pmatrix} =  
    n \begin{pmatrix}
        \bbeta^{\top} \bar{\bx}_1  \\
        \vdots \\
        \bbeta^{\top} \bar{\bx}_n
    \end{pmatrix} = n \bar{\bX}\bbeta.
\end{equation}
    Similarly, $(\sum_{k=1}^p \bX_k \tilde{\beta}_k) \mathbf{1}_n = n \bar{\bX}\tilde{\bbeta}$. Multiplying both sides of 
    Equation 
    (\ref{eq:identify}) on the right by $\mathbf{1}_n$, applying 
    Equation (\ref{eq:one_mul}), and rearranging, we have
\[
\bar{\bX} (\bbeta - \tilde{\bbeta}) = \mathbf{0}_n.
\]
    Next, multiplying both sides of the previous expression on the left by $\bar{\bX}^{\top}$, we have
\[
    \bar{\bX}^{\top} \bar{\bX} (\bbeta - \tilde{\bbeta}) = \mathbf{0}_p.
\]
    Since $\bar{\bX}^{\top} \bar{\bX}$ is full rank, we must have $\bbeta = \tilde{\bbeta}$. From the previous equality, the fact that $\bU \bLambda \bU^{\top} = \tilde{\bU} \tilde{\bLambda} \tilde{\bU}^{\top}$ follows immediately from Equation (\ref{eq:identify}).
\end{proof}

\begin{proof}[Proof of Proposition~\ref{prop:ssibp-ordering}]

The proof proceeds directly:
\begin{align*}
\bbP(&\abs{\eta_h - \eta_0} \leq \epsilon) = \bbE[\bbP(\bbB_{\epsilon}(\eta_0) \mid \theta_h)] \\
&= \bbE(\theta_h) \bbP_{slab}(\bbB_{\epsilon}(\eta_0)) + \bbE(1 - \theta_h) \bbP_{spike}(\bbB_{\epsilon}(\eta_0)) \\
&= \left(\frac{a}{a + \kappa + 1}\right) \left(\frac{a}{a + 1}\right)^{h-1} \bbP_{slab}(\bbB_{\epsilon}(\eta_0)) +  \left[1 - \left(\frac{a}{a + \kappa + 1}\right)\left(\frac{a}{a + 1}\right)^{h-1} \right] \bbP_{spike}(\bbB_{\epsilon}(\eta_0)) \\
&= \bbP_{spike}(\bbB_{\epsilon}(\eta_0)) + \left(\frac{a}{a + \kappa + 1}\right)\left(\frac{a}{a + 1}\right)^{h-1} \left[\bbP_{slab}(\bbB_{\epsilon}(\eta_0)) - \bbP_{spike}(\bbB_{\epsilon}(\eta_0))\right].
\end{align*}
Since $[a/(a + 1)]^{h-1}$ is decreasing in $h$, the result follows.
\end{proof}

\begin{proof}[Proof of Proposition~\ref{prop:prior_support}]

Before proving the result, we need to formally define the support of a probability measure and prove a general measure theoretic lemma.

\begin{definition}
Let $\bbP$ be a probability measure on a measurable space $(X, \mathcal{B}(X))$, where $X$ is a topological space and $\mathcal{B}(X)$ is the associated Borel $\sigma$-algebra. The support of $\bbP$, denoted by $\supp(\bbP)$, is the set of all $x \in X$ such that every open neighborhood of $x$ has non-zero probability.
\end{definition}

\begin{lemma}\label{lemma:support}
Let $(X_1, \mathcal{B}(X_1))$ and $(X_2, \mathcal{B}(X_2))$ be measurable spaces where $X_1$ and $X_2$ are topological spaces and $\mathcal{B}(X_1)$ and $\mathcal{B}(X_2)$ are their associated Borel $\sigma$-algebras. Also, let $\bbP$ be a probability measure on $X_1$ and $f: X_1 \rightarrow X_2$ be a continuous, surjective function. If $\supp(\bbP) = X_1$, then $\bbP_f$, the push-forward measure on $X_2$ associated with $f$ defined as $\bbP(f^{-1}(B))$ for all $B \in \mathcal{B}(X_2)$, has $\supp(\bbP_f) = X_2$.
\end{lemma}

\begin{proof}
First note that the continuity of $f$ implies that $f$ is $\mathcal{B}(X_1)/\mathcal{B}(X_2)$-measurable, so the statements made above are well defined. Now pick an arbitrary $x \in X_2$. Note that for any open neighborhood $N_x$ of $x$, $f^{-1}(N_x) \neq \emptyset$ because $f$ is surjective. In addition, since $f^{-1}(N_x)$ is a non-empty open set in $X_1$, it is also an open neighborhood of some point in $X_1$. Therefore, $\bbP_f(N_x) = \bbP(f^{-1}(N_x)) > 0$ for any $N_x$ by the assumption that $\bbP$ has full support on $X_1$. Since $x$ was arbitrary, this proves that $\bbP_f$ has full support on $X_2$.
\end{proof}

We are now ready to prove Proposition~\ref{prop:prior_support}. First, we show that the function $f: \Reals{n \times d} \rightarrow \bar{\mathcal{V}}_{d,n}$ defined by steps 2 and 3 in Section~\ref{subsec:semi_ortho_prior} is surjective. Pick an arbitrary $\bA \in \bar{\mathcal{V}}_{d,n}$. Let $\bB = \bA$, then $\bar{\bB} = [\mathbf{1}_n, \bA] \ \bI_{d+1}$ so that $\bQ = [\mathbf{1}_n, \bA]$ and $\bR = \bI_{d+1}$. Therefore, $\bU = \bQ_{2:(d+1)} = \bA$, which proves $f$ is surjective. Furthermore, $f$ is differentiable with respect to $\bB$ and thus continuous. The result follows from Lemma~\ref{lemma:support} with $X_1 = \Reals{n \times d}$, $X_2 = \bar{\mathcal{V}}_{d,n}$, and $\bbP$ the Gaussian measure on $\Reals{n \times d}$ (that is, the elements of $X_1$ are independently drawn from a $N(0, 1)$ distribution) which has full support.
\end{proof}

\subsection{Proof of Theorem~\ref{thm:exp_decay} and Theorem~\ref{thm:post_con}}\label{sec:post_con}

Recall that we model the observed edge variables using a GLNEM.
In particular, we assume the edge variables, $Y_{ij}$, are independently drawn from an exponential dispersion family with densities of the form:
\[
q(y_{ij}; \theta_{ij}, \phi) = q_{ij}(y_{ij}) = \exp\left\{\frac{y_{ij} \, \theta_{ij} - b(\theta_{ij})}{\phi} + k(y_{ij}, \phi) \right\}, \qquad 1 \leq i \leq j \leq n,
\]
where
\begin{equation}
    g(b'(\theta_{ij})) = \bbeta^{\top}\bx_{ij} + \left[\bU \bLambda \bU^{\top}\right]_{ij}
\end{equation}
for a strictly increasing link function $g$. We will refer to the right hand side of the above equation as $\eta_{ij} = \bbeta^{\top}\bx_{ij} + \left[\bU \bLambda \bU^{\top}\right]_{ij}$. We will also write $\xi = (g \circ b')^{-1}$, so that $\theta_{ij} = \xi(\eta_{ij})$.

In addition, we assume that the matrices $\bY^{(n)}$ are generated from
model (\ref{eq:random}) with true parameters $\set{\bbeta_0, \bLambda_0, \bU_0}$, where $\bbeta_0 \in \Reals{p}$, $\bLambda_0 \in \Reals{d_0 \times d_0}$, $\bU_0 \in \mathcal{V}_{d_0, n}$, and $\norm{\bLambda_0}_0 = d_0$. Furthermore, let 
\begin{equation}
    \eta_{0ij} = \bbeta_0^{\top} \bx_{ij} + \left[\bU_{0} \bLambda_0 \bU_{0}^{\top}\right]_{ij},
\end{equation}
and $\theta_{0ij} = \xi(\eta_{0ij})$. Let $\bbP_0^{(n)}$ and $\bbE_0^{(n)}$ denote the probability and expectation under this true data generating process. Finally, we denote the exponential dispersion family density under this true data generating process as
\[
q(y_{ij}; \theta_{0ij}, \phi) = q_{0ij}(y_{ij}) = \exp\left\{\frac{y_{ij} \, \theta_{0ij} - b(\theta_{0ij})}{\phi} + k(y_{ij}, \phi) \right\}, \qquad 1 \leq i \leq j \leq n.
\]

Our proofs of Theorem~\ref{thm:exp_decay} and Theorem~\ref{thm:post_con} make use of the theory developed by \citetSup{goshal2007} and \citetSup{jeong2021} to demonstrate that the posterior concentrates around low-dimensional models for GLNEMs.
Furthermore, we synthesize recent theoretical results for sparse Gaussian factor models~\citepSup{rockova2016, ohn2021} to derive our result on the $\SSIBP$ prior's ability to concentrate on low-dimensional models and derive our posterior concentration rates.

\begin{proof}[Proof of Theorem~\ref{thm:exp_decay}]
Since the inequality follows trivially when $t \geq d$, we assume $t < d$. Define the random variable $S_d = \norm{\boldeta}_0 = \sum_{h=1}^d \ind{\eta_h \neq 0}$. We bound $\bbP(S_d > t)$ as follows:
\[
\bbP(S_d > t) \leq \bbE\left[\bbP(S_d > t \mid \set{\theta_h}_{h=1}^d) \ind{\theta_1 < \tilde{\theta}}\right] + \bbP(\theta_1 \geq \tilde{\theta}),
\]
where $\tilde{\theta} = t \delta \log d / 6 d^{1 + \delta}$. Conditioned on $\set{\theta_h}_{h=1}^d$, $S_d$ is the sum of independent Bernoulli random variables with success probabilities
\[
p_h := \bbP(\lambda_h \neq 0 \mid \set{\theta_h}_{h=1}^d) = \theta_h,
\]
with the property $p_1 > p_2 > \dots > p_d$ by construction. In addition, on the event $\set{\theta_1 < \tilde{\theta}}$, we have $p_1 = \theta_1 < t \delta \log d / 6 d^{1+\delta} = (t/6d) \,  \delta \log d \, \exp(-\delta \log d) < t / d < 1$, where we used the fact that $e^{-x} < 1 / x$ for $x > 0$. As such, we can apply Lemma~\ref{lemma:chernoff}:
\begin{align}\label{eq:sum_bound}
\bbP(S_d > t \mid \set{\theta_h}_{h=1}^d) &= \bbP(S_d > \frac{t}{d} d \mid \set{\theta_h}_{h=1}^d) \nonumber \\
&\leq \left[\left\{\theta_1 \frac{d}{t}\right\}^{t/d} e^{t/d}\right]^d \nonumber \\
&\leq \left\{\frac{t \delta \log d}{6 d^{1 + \delta}} \frac{d}{t}\right\}^t e^t \\
&= \left\{\frac{\log d^{\delta/2}}{3 d^{\delta}}\right\}^t e^t, \nonumber \\
&\leq e^{-t(\delta/2) \log d}, \nonumber
\end{align}
where we used the fact that $\log x \leq (1/e) x$ for any $x > 0$ in the fifth line.

Now we bound $\bbP(\theta_1 \geq \tilde{\theta}) = \bbP(\nu_1 \geq \tilde{\theta})$, where $\nu_1 \sim \Beta(a, \kappa + 1)$ with $\kappa = d^{1 + \delta}$. Since $0 < a < 1$, Gautschi's inequality~\citepSup{gautschi1959} implies
\[
B(a, \kappa + 1) = \frac{\Gamma(a)}{\kappa + a} \frac{\Gamma(\kappa + 1)}{\Gamma(\kappa + a)} > \frac{\Gamma(a)}{\kappa + 1} \kappa^{1 - a} \geq \frac{\kappa^{1 - a}}{\kappa + 1}.
\]
The last line of the previous display used the fact that $\Gamma(a) \geq 1$ for $a \in (0, 1)$, which follows from $\Gamma(a)$ being monotonically decreasing on $(0, 1)$ with $\Gamma(1) = 1$. Then we have
\begin{align*}
\bbP(\nu_1 > \tilde{\theta}) &= \frac{1}{B(a, \kappa + 1)} \int_{\tilde{\theta}}^1 \nu^{a - 1} (1 - \nu)^{\kappa} d\nu \\
&\leq \frac{\kappa + 1}{\kappa^{1 - a}} \, \tilde{\theta}^{a - 1} \int_{\tilde{\theta}}^1 (1 - \nu)^{\kappa} d\nu \\
&= \left(\frac{1}{\kappa \tilde{\theta}}\right)^{1 - a}\left(1 - \frac{t \delta \log d}{6 d^{1 + \delta}}\right)^{d^{1 + \delta} + 1} \\
&= \left(\frac{6}{t\delta \log d}\right)^{1 - a}\left(1 - \frac{t \delta \log d}{6 d^{1 + \delta}}\right)^{d^{1 + \delta} + 1} \\
&\leq  \exp\left\{(a - 1) \log(t \delta \log d / 6)\right\} \ \exp\left(-\frac{t \delta}{6} \log d\right),
\end{align*}
where we used the fact that $(1 - b/x)^x \leq e^{-b}$ for $b \in \Reals{}$ and $x > 0$ in the last line. Next, we note that $t\delta \log d / 6 \geq \delta \log d / 6  \geq 1$ by the assumption that $\delta \geq 6 / \log d$ and $t \geq 1$. It follows that the first term in the last line of the previous display is bounded above by 1, so that
\begin{equation}\label{eq:comp_bound}
\bbP(\nu_1 > \tilde{\theta}) \leq \exp\left(-t (\delta/6) \log d\right).
\end{equation}
Combining (\ref{eq:sum_bound}) and (\ref{eq:comp_bound}), we obtain the desired bound.
\end{proof}

Before proving Theorem~\ref{thm:post_con} we need the following Lemma.

\begin{lemma}\label{lemma:elbo}
Suppose that the conditions of Theorem \ref{thm:post_con} hold. Then $\mathbb{P}_0^{(n)}(\mathcal{A}_n^c) = o(1)$, where the set $\mathcal{A}_n$ is
\[
\mathcal{A}_n = \left\{Y_{ij} \in \set{0, 1}, \ 1 \leq i \leq j \leq n \ : \ \int \int \int \prod_{i\leq j} \frac{q_{ij}(Y_{ij})}{q_{0ij}(Y_{ij})} d\Pi(\bbeta) d\Pi(\bLambda) \textrm{d}\Pi(\bU)  \geq e^{- C_1 n (n-1) \epsilon_n^2}\right\}
\]
for
\[
    \epsilon_n =  K_{\lambda} \, \sqrt{\frac{d_0 \log d}{n(n-1)}},
\]
and some constant $C_1 > 1$ that only depends on $\delta$.
\end{lemma}
\begin{proof}
Denote the KL-neighborhood for the model parameters by
\[
\mathcal{B}_n = \left\{ (\bbeta, \bLambda, \bU) : \ \frac{2}{n(n-1)} \sum_{i\leq j} K(q_{0ij}, q_{ij}) \leq \epsilon_n^2,\ \frac{2}{n(n-1)} \sum_{i \leq j} V(q_{0ij}, q_{ij}) \leq \epsilon_n^2 \right\}.
\]
By Lemma 10 of \citetSup{goshal2007},  we have for any $C > 0$,
\[
\bbP_0^{(n)}\left\{\int_{\mathcal{B}_n} \prod_{i \leq j} \frac{q_{ij}(Y_{ij})}{q_{0ij}(Y_{ij})} d\Pi(\bbeta) d\Pi(\bLambda) d\Pi(\bU) \geq e^{-(1 + C)n(n-1)\epsilon_n^2} \bbP(\mathcal{B}_n)\right\} \geq 1 - \frac{1}{C^2 n(n-1)\epsilon_n^2}.
\]

Therefore, it suffices to show that the prior probability $\bbP(\mathcal{B}_n) \geq e^{-C_2 n(n-1) \epsilon_n^2}$ for some constant $C_2 > 0$. By Taylor expanding $K(q_{0ij}, q_{ij})$ and $V(q_{0ij}, q_{ij})$, Lemma 1 of \citetSup{jeong2021} established that 
\[
    \max\left\{K(q_{0ij}, q_{ij}), V(q_{0ij}, q_{ij})\right\} \leq \frac{b''(\theta_{0ij})[\xi'(\eta_{0ij})]^2}{\phi} (\eta_{ij} - \eta_{0ij})^2 + o\left\{(\eta_{ij} - \eta_{0ij})^2\right\}.
\]

    First, we show that $\abs{\theta_{0ij}}$ is uniformly bounded above by a constant for all $1 \leq i,j \leq n$. Specifically, $\abs{\theta_{0ij}} = \abs{\xi(\eta_{0ij})} \leq \max_{i,j} \abs{\xi(\bbeta^{\top}_0 \bx_{ij} + \bu_{0i}^{\top} \bLambda_0 \bu_{0j})} \leq \abs{\xi(K_{\beta} K_x + K_{\lambda} K_u^2)} < \infty$, since $\xi$ is continuous and strictly increasing and we used the parameter bounds in Assumptions \ref{assump:A4} -- \ref{assump:A5}. Coupled with Assumption \ref{assump:A7}, we have that $b''(\theta_{0ij})$ is uniformly bounded above by a constant. Next, using the chain-rule, we have that $\xi'(\eta_{0ij}) = (g^{-1})'(\eta_{0ij}) / b''(\theta_{0ij}) \leq \sup_{\set{\eta:\ \abs{\eta} \leq K_{\beta} K_x + K_{\lambda} K_u^2}} (g^{-1})'(\eta) / K_{b,1} \leq K_g / K_{b,1}$, where we used the lower bound in Assumption \ref{assump:A7} and Assumption \ref{assump:A8}. Therefore, both $K(q_{0ij}, q_{ij})$ and $V(q_{0ij}, q_{ij})$ can be bounded above by a constant multiple of $(\eta_{ij} - \eta_{0ij})^2$ for sufficiently large $n$. Thus, we have for some constant $b_1 > 0$ and sufficiently large $n$,
\begin{align*}
\bbP(\mathcal{B}_n) &\geq \bbP\left(\sum_{i \leq j} (\eta_{ij} - \eta_{0ij})^2 \leq b_1 n(n-1) \epsilon_n^2\right) \nonumber \\
&\geq \bbP\left(\sum_{i=1}^n \sum_{j=1}^n (\eta_{ij} - \eta_{0ij})^2 \leq b_1 n (n-1) \epsilon_n^2\right) \nonumber \\
    &= \bbP\left(\norm{\sum_{k=1}^p \bX_k \beta_k + \bU \bLambda \bU^{\top} - \sum_{k=1}^p \bX_k \beta_{0k} - \bU_0 \bLambda_0 \bU_0^{\top}}_F^2 \leq b_1 n (n - 1) \epsilon_n^2\right) \nonumber \\
    &\geq \bbP\left(\norm{\sum_{k=1}^p \bX_k (\beta_k - \beta_{0k})}_F\leq \sqrt{\frac{b_1n(n-1)}{4}} \epsilon_n \right) \\
    &\qquad\qquad\qquad \bbP\left(\norm{\bU \bLambda \bU^{\top} - \bU_0 \bLambda_0 \bU_0^{\top}}_F \leq \sqrt{\frac{b_1n(n-1)}{4}} \epsilon_n\right),
\end{align*} 
where the last line follows from the triangle inequality and the fact that $\bbeta$ is independent of the remaining variables under our prior specification.

Let $\bU_{0,h} \in \Reals{n}$ and $\bU_h \in \Reals{n}$ denote the $h$-th column of $\bU_0$ and $\bU$, respectively. Furthermore, we define $\tilde{\bLambda}_0 = \diag(\tilde{\blambda}_0) \in \Reals{d \times d}$ and $\tilde{\bU}_0 \in \Reals{n \times d}$ such that $\tilde{\lambda}_{0h} = \lambda_{0h}$ and $\tilde{\bU}_{0,h} = \bU_{0,h}$ for $h \leq d_0$ and $\tilde{\lambda}_{0h} = 0$ and $\tilde{\bU}_{0,h} = \mathbf{0}_n$ for $h > d_0$. In particular, this implies $[\tilde{\bU}_{0,1}, \dots, \tilde{\bU}_{0, d_0}] \in \mathcal{V}_{d_0, n}$. Next, we have
\begin{align*}
    \norm{\bU \bLambda \bU^{\top} - \bU_0 \bLambda_0 \bU_0^{\top}}_F &= \norm{\bU \bLambda \bU^{\top} - \tilde{\bU}_0 \tilde{\bLambda}_0 \tilde{\bU}_0^{\top}}_F \\
    &=\norm{\bU \bLambda \bU^{\top} - \bU \tilde{\bLambda}_0 \bU^{\top} + \bU \tilde{\bLambda}_0 \bU^{\top} - \tilde{\bU}_0 \tilde{\bLambda}_0 \tilde{\bU}_0^{\top}}_F \\
    &\leq \norm{\bU \bLambda \bU^{\top} - \bU \tilde{\bLambda}_0 \bU^{\top}}_F + \norm{\bU \tilde{\bLambda}_0 \bU^{\top} - \tilde{\bU}_0 \tilde{\bLambda}_0 \tilde{\bU}_0^{\top}}_F \\
    &= \norm{\bLambda - \tilde{\bLambda}_0}_F + \norm*{\sum_{h=1}^{d_0} \lambda_{0h} \left(\bU_h \bU^{\top}_{h} - \bU_{0,h} \bU_{0,h}^{\top}\right)}_F \\
    &\leq \norm{\bLambda - \tilde{\bLambda}_0}_F + \sum_{h=1}^{d_0} \abs{\lambda_{0h}} \ \norm{\bU_h \bU^{\top}_{h} - \bU_{0,h} \bU_{0,h}^{\top}}_F \\
    &\leq \norm{\bLambda - \tilde{\bLambda}_0}_F + \norm{\blambda_0}_{\infty} \sum_{h=1}^{d_0} \norm{\bU_h \bU^{\top}_{h} - \bU_{0,h} \bU_{0,h}^{\top}}_F.
\end{align*}
    Combining this with the previous inequality and use Assumptions \ref{assump:A4} and \ref{assump:A5}, we have
\begin{align}
    \bbP(\mathcal{B}_n) &\geq \bbP\left( \norm{\bbeta - \bbeta_0}_2 \leq \sqrt{\frac{b_1 n (n - 1)}{4 n^2 K_x^2}} \epsilon_{n} \right) \bbP\left(\norm{\bLambda - \tilde{\bLambda}_0}_F \leq \sqrt{\frac{b_1 n(n-1)}{16}} \epsilon_n\right) \times  \nonumber \\
    &\qquad\bbP\left( \sum_{h=1}^{d_0} \norm{\bU_{h} \bU_{h}^{\top} - \bU_{0,h} \bU_{0,h}^{\top}}_F \leq \sqrt{\frac{b_1 n(n-1)}{16}} \frac{\epsilon_{n}}{K_{\lambda}}\right), \label{eq:lower_bound}
\end{align}
where we used the fact that $\bLambda$ and $\bU$ are independent under the prior. Now we bound each term 
on the right hand side of (\ref{eq:lower_bound}) from below.

    We start with the last term
    \begin{align}
    \bbP\bigg(\sum_{h=1}^{d_0} &\norm{\bU_{h} \bU_{h}^{\top} - \bU_{0,h} \bU_{0,h}^{\top}}_F \leq \sqrt{\frac{b_1 n (n-1)}{16}} \frac{\epsilon_n}{K_{\lambda}}\bigg) = \nonumber \\
    &\bbP\left(\frac{1}{\sqrt{2}} \sum_{h=1}^{d_0} \norm{\bU_{h} \bU_{h}^{\top} - \bU_{0,h} \bU_{0,h}^{\top}}_F \leq \sqrt{\frac{b_1 d_0 \log d}{32}}\right). \label{eq:u_bound}
\end{align}
    We assert that $\sum_{h=1}^{d_0} \norm{\bU_h \bU_h^{\top} - \bU_{0,h} \bU_{0,h}^{\top}}_F/ \sqrt{2} \leq d_0$. Indeed,
\begin{align*}
    \sum_{h=1}^{d_0} \norm{\bU_h\bU_h^{\top} - \bU_{0,h}\bU_{0,h}^{\top}}_F &= \sum_{h=1}^{d_0} \left[\norm{\bU_h\bU_h^{\top}}_F^2 + \norm{\bU_{0,h}\bU_{0,h}^{\top}}_F^2 - 2 \tr(\bU_h\bU_h^{\top}\bU_{0,h}\bU_{0,h}^{\top})\right]^{1/2} \\
    &= \sum_{h=1}^{d_0} \left[2 - 2 (\bU_h^{\top} \bU_{0,h})^2\right]^{1/2} \\
&\leq \sqrt{2} d_0.
\end{align*}
    Therefore, the probability on the right hand side of 
    (\ref{eq:u_bound}) equals one for sufficiently large $n$ by Assumption \ref{assump:A2}.

    Next, we apply Lemma~\ref{lemma:ssibp_lb} to bound the second term in 
    (\ref{eq:lower_bound}) from below. Note that we satisfy the condition that $d_0 \leq d/\sqrt{2}$ for sufficiently large $n$ by Assumption \ref{assump:A2}, so the lemma implies that
\begin{align*}
    \bbP\left(\norm{\bLambda - \tilde{\bLambda}_0}_F \leq \sqrt{\frac{b_1n(n-1)}{16}} \epsilon_n\right) &= \bbP\left(\norm{\blambda - \tilde{\blambda}_0}_2 \leq K_{\lambda} \sqrt{b_1 d_0 \log d/16}\right) \\
    &\geq e^{-\tilde{C} d_0 \log\left[bd \, \vee \, (K_{\lambda}\sqrt{ b_1 d_0 \log d/16})^{-1}\right] - K_{\lambda} \sqrt{b_1 d_0 \log d/16} - K_{\lambda} d_0} \\
    &\geq e^{-C_3 K_{\lambda}^2 d_0 \log bd} \\
    &\geq e^{-C_4 K_{\lambda}^2 d_0 \log d} \\
    &= e^{-C_4 n(n-1) \epsilon_n^2}
\end{align*}
    for some constants $C_3$ and $C_4$ that depends on $\delta$. In the third line of the previous display, we used the fact that $bd \, \vee \, (K_{\lambda}\sqrt{b_1 d_0 \log d/16})^{-1} = bd$ for sufficiently large $n$ because $b \geq 1$. Also, we used Condition \ref{assump:A1} in the last line of the previous display.

Lastly, we provide a lower bound for the first term in 
(\ref{eq:lower_bound}). First, we have that $\sqrt{b_1 n (n-1) / (4 n^2 K_x^2)} \geq b_3$ for some constant $b_3$ and $n$ sufficiently large. Thus,
\begin{align*}
\bbP\left( \norm{\bbeta - \bbeta_0}_2 \leq \sqrt{\frac{b_1 n(n-1)}{4 n^2 K_x^2}} \epsilon_{n} \right) &\geq \bbP\left(\norm{\bbeta - \bbeta_0}_2 \leq b_3 \epsilon_n\right)  \\
&\geq \prod_{k=1}^p \bbP\left(\abs{\beta_k - \beta_{0k}} \leq \frac{b_3}{\sqrt{p}}\epsilon_n\right).
\end{align*}
For simplicity, let $b_4 = b_3 / \sqrt{p}$. Since the prior density on $\beta_k \iidsim N(0, \sigma^2_{\beta})$ is bounded away from zero in a neighborhood of $\beta_{0k}$ for $1 \leq k \leq p$, we have
\begin{align*}
\bbP(\abs{\beta_k - \beta_{0k}} \leq b_4 \epsilon) &= (2\pi \sigma_{\beta}^2)^{-1/2} \int_{\abs{\beta_k - \beta_{0k}} \leq b_4 \epsilon_n} e^{-\beta_k^2/2\sigma_{\beta}^2} d\beta_k \\
&\geq m_k (2\pi \sigma_{\beta}^2)^{-1/2} \int_{\abs{\beta_k - \beta_{0k}} \leq b_4 \epsilon_n} d\beta_k \\
&= \tilde{b}_k \epsilon_n,
\end{align*}
where $\tilde{b}_k = 2 b_4 m_k (2\pi\sigma_{\beta}^2)^{-1/2}$ and $m_k > 0$ is the lower bound on the Gaussian density in the neighborhood of $\beta_{0k}$. Combining this result with the previous expression, we have that there is a constant $b_4$ such that for large enough $n$,
\begin{align*}
\bbP\left( \norm{\bbeta - \bbeta_0}_2 \leq \sqrt{\frac{b_1 n(n-1)}{4 n^2 K_x^2}} \epsilon_{n} \right) &\geq b_4 \epsilon_n^p \\
&= e^{-(p/2)\log(n(n-1)) + p \log K_{\lambda} + (p/2)\log(d_0 \log d) + \log b_4} \\
&\geq e^{-b_5 p\log(n)} \\
&\geq e^{-b_5 (p/\gamma) \log d} \\
&\geq e^{-C_5 n(n-1) \epsilon_n^2},
\end{align*}
for some constants $b_5 > 0$ and $C_5 > 0$, where we used Condition \ref{assump:A3} on the second to last line.

Combining the previous three bounds for the probabilities in 
(\ref{eq:lower_bound}) with $C_2 = \max(C_4, C_5)$ gives the desired bound on $\bbP(\mathcal{B}_n)$. The result follows with $C_1 = 1 + C + C_2$.
\end{proof}

\begin{proof}[Proof of Theorem \ref{thm:post_con}]
As in \citetSup{jeong2021}, we write the posterior as
\begin{equation}\label{eq:posterior}
\bbP(\norm{\bLambda}_0 > C d_0 \mid \bY) = \frac{\int \int \int_{\norm{\bLambda}_0 > C d_0} \prod_{i \leq j} \frac{q_{ij}(Y_{ij})}{q_{0ij}(Y_{ij})}d\Pi(\bLambda) d\Pi(\bbeta) d\Pi(\bU)}{\int\int\int \prod_{i \leq j} \frac{q_{ij}(Y_{ij})}{q_{0ij}(Y_{ij})} d\Pi(\bLambda)d\Pi(\bbeta)d\Pi(\bU)} = \frac{N_n}{D_n},
\end{equation}
where $\Pi(\bLambda)$, $\Pi(\bbeta)$, and $\Pi(\bU)$ denote the prior distributions for $\bLambda, \bbeta$, and $\bU$, respectively. Now, we introduce an event $\mathcal{A}_n$ with large probability under the true data generating process. In particular, let $\mathcal{A}_n = \set{\bY \ : \ \int \int \int \prod_{i \leq j} \frac{q_{ij}(Y_{ij})}{q_{0ij}(Y_{ij})} d\Pi(\bLambda) d\Pi(\bbeta) d\Pi(\bU) \geq e^{-C_1 n (n-1) \epsilon_n^2}}$, where $\epsilon_n$ is the rate in Lemma \ref{lemma:elbo}. If we decompose the probability in Equation (\ref{eq:posterior}) into the sum of two complementary conditional probabilities (conditioning on $\mathcal{A}_n$ and its complement $\mathcal{A}_n^c$), we can write:
\begin{align*}
\mathbb{E}_0^{(n)}\left[\bbP(\norm{\bLambda}_0 > C d_0 \mid \bY)\right] &\leq \mathbb{E}_0^{(n)}\left[\bbP(\norm{\bLambda}_0 > C d_0 \mid \bY) \mathbbm{1}_{\mathcal{A}_n}\right] + \mathbb{P}_0^{(n)}(\mathcal{A}_n^c).
\end{align*}
     On the event $\mathcal{A}_n$, $D_n$ is bounded below by $e^{-C_1 K_{\lambda}^2 d_0 \log d}$. On the other hand, the expected value of the numerator can by bounded above by $\bbP(\norm{\bLambda}_0 > C d_0)$ using Fubini's theorem:
\begin{align*}
\bbE_0^{(n)}&\left[\int \int \int_{\norm{\bLambda}_0 > C d_0} \prod_{i \leq j} \frac{q_{ij}(Y_{ij})}{q_{0ij}(Y_{ij})}d\Pi(\bLambda) d\Pi(\bbeta) d\Pi(\bU)\right] = \\
&\qquad \int \int \int_{\norm{\bLambda}_0 > C d_0} \prod_{i \leq j} \bbE_0^{(n)}\left[\frac{q_{ij}(Y_{ij})}{q_{0ij}(Y_{ij})}\right]d\Pi(\bLambda) d\Pi(\bbeta) d\Pi(\bU) \leq \bbP(\norm{\bLambda}_0 > C d_0).
\end{align*}
Therefore, we can use Lemma~\ref{lemma:elbo} to conclude that
\begin{align*}
    \mathbb{E}_0^{(n)}\left[\bbP(\norm{\bLambda}_0 > C d_0 \mid \bY)\right] &\leq \bbP(\norm{\bLambda}_0 > C d_0) e^{C_1 K_{\lambda}^2 d_0 \log d} + o(1).
\end{align*}
In addition, for $n$ sufficiently large we have that $\delta > 6 / \log d$, since $d = \lfloor n^{\gamma}\rfloor$ by Condition \ref{assump:A3}, so we can apply Theorem~\ref{thm:exp_decay} to bound the previous expression from above as
\begin{align*}
    \mathbb{E}_0^{(n)}\left[\bbP(\norm{\bLambda}_0 > C d_0 \mid \bY)\right] \leq 3e^{-(C(\delta/6) - C_1 K_{\lambda}^2) d_0 \log d} + o(1),
\end{align*}
    for $n$ large enough, which goes to zero as $n \rightarrow \infty$ for $C > C_1 (6/\delta) K_{\lambda}^2 \, \vee \, 1$.

\end{proof}

\subsection{Auxiliary Lemmas}\label{sec:aux_lemma}

\begin{lemma}\label{lemma:ssibp_lb}
Let $\blambda \sim \text{SS-IBP}_d(a, d^{1+\delta}, \bbP_{spike}, \bbP_{slab})$ with $\delta > 0$, $\bbP_{slab} = \Laplace(b)$ and $\bbP_{spike} = \delta_0$. If $\blambda_0 \in \Reals{d}$ has nonzero support on the first $d_0$ entries with $1 \leq d_0 \leq d/\sqrt{2}$, i.e., $\lambda_{0j} \neq 0$ for $j \leq d_0$ and $\lambda_{0j} = 0$ for $j > d_0$, then for any $\epsilon > 0$,
\[
\bbP\left(\norm{\blambda - \blambda_0}_2 \leq \epsilon \right) \geq  e^{-\epsilon/\sqrt{2} b} e^{-\norm{\lambda_0}_1/b} a^{d_0} e^{-C_1 d_0 \log(bd \, \vee \, \epsilon^{-1})},
\]
for some positive constant $C_1$ that only depends on $\delta$. Moreover, if $\norm{\blambda_0}_{\infty} \leq M$ for some $M > 0$, $b \geq 1$, and $a = 1/d$, we have
\[
\bbP\left(\norm{\blambda - \blambda_0}_2 \leq \epsilon \right) \geq e^{-C_2 d_0 \log(bd \, \vee \, \epsilon^{-1}) - \epsilon - d_0 M},
\]
for some positive constant $C_2$ that depends only on $\delta$.
\end{lemma}
\begin{proof}
Conditioned on $\set{\theta_h}_{h=1}^d$, it follows from Lemma~\ref{lemma:ss_laplace} that
\begin{align*}
\bbP(\norm{\blambda - \blambda_0}_2 &\leq \epsilon \mid \set{\theta_h}_{h=1}^d) \geq \theta_{d_0}^{d_0} (1 - \theta_1)^{d - d_0} \left[e^{-\norm{\blambda_0}_1/b - \epsilon/\sqrt{2}b} \left(\frac{\epsilon}{\sqrt{2} bd_0}\right)^{d_0}\right] \\
&\geq e^{-\norm{\blambda_0}_1/b - \epsilon/\sqrt{2}b - 2 d_0 \log(\sqrt{2} bd_0 \, \vee \, \epsilon^{-1})} \left[\theta_{d_0}^{d_0} (1 - \theta_1)^{d - d_0} \right] \\
&\geq e^{-\norm{\blambda_0}_1/b - \epsilon/\sqrt{2}b - 2 d_0 \log(bd \, \vee \, \epsilon^{-1})} \left[\theta_{d_0}^{d_0} (1 - \theta_1)^{d - d_0}\right],
\end{align*}
where in the last line we used $\sqrt{2} d_0 \leq d$. Focusing on the terms involving $\set{\theta_h}_{h=1}^d$, we have
\begin{align*}
\theta_{d_0}^{d_0}(1 - \theta_1)^{d - d_0} &= \left(\prod_{h=1}^{d_0} \nu_h\right)^{d_0} (1 - \nu_1)^{d - d_0} \\
&=\left(\prod_{h=2}^{d_0} \nu_h^{d_0}\right) \, \nu_1^{d_0}(1 - \nu_1)^{d - d_0},
\end{align*}
so that
\[
\bbE\left[\theta_{d_0}^{d_0}(1 - \theta_1)^{d - d_0}\right] = \left[\bbE(\nu_2^{d_0})\right]^{d_0 - 1} \, \bbE\left[\nu_1^{d_0}(1 - \nu_1)^{d - d_0}\right],
\]
due to the independence of $\nu_1, \dots, \nu_d$ and the fact that $\nu_2, \dots, \nu_d$ are identically distributed. Since $B(a, d^{1+\delta} + 1) < B(a, 1) = 1 / a$, $a - 1 < 0$, and $d^{1 + \delta} > d - d_0$, it follows that
\begin{align*}
\bbE\left[\nu_1^{d_0}(1 - \nu_1)^{d - d_0}\right] &= \frac{1}{B(a, d^{1 + \delta} + 1)} \int_0^1 \nu_1^{d_0 + a - 1}(1 - \nu_1)^{d-d_0 + d^{1 + \delta}} d\nu_1 \\
&\geq a \left(1 - \frac{1}{d^{1+\delta}}\right)^{2d^{1+\delta}} \int_0^{d^{-(1 + \delta)}} \nu_1^{d_0} d\nu_1 \\
&\geq \frac{a}{(d_0 + 1)} e^{-4} \left(\frac{1}{d^{1 + \delta}}\right)^{d_0 + 1} \\
&\geq a e^{-C_1d_0\log d}
\end{align*}
for some $C_1 > 0$ depending only on $\delta$, where the second inequality used the fact that $(1 - x)^{1/x} \geq e^{-2}$ for $0 < x < 1/2$. Similarly we have
\begin{align*}
\bbE\left(\nu_2^{d_0}\right) &= \frac{1}{B(a, 1)} \int_{0}^1 \nu_2^{d_0 + a - 1} d\nu_2 \\
&\geq a \int_0^1 \nu_2^{d_0} d\nu_2 \\
&= \frac{a}{d_0 + 1} \\
&\geq a e^{-\log d},
\end{align*}
where the second line used the fact that $a < 1$. Combining these results, we recover the desired bound:
\[
\bbP(\norm{\blambda - \blambda_0}_2 \leq \epsilon) = \bbE\left[\bbP(\norm{\blambda - \blambda_0}_2 \leq \epsilon \mid \set{\theta_h}_{h=1}^d)\right] \geq e^{-\epsilon/\sqrt{2}b} e^{-\norm{\blambda_0}_1/b} a^{d_0} e^{-C_4 d_0 \log (bd \, \vee \, \epsilon^{-1})},
\]
where $C_4$ is a positive constant that only depends on $\delta$.

For the second assertion, note that $\norm{\blambda_0}_{\infty} \leq M$ implies
\[
\norm{\blambda_0}_1 \leq d_0 \norm{\blambda_0}_{\infty} \leq d_0 M,
\]
which completes the proof.
\end{proof}

\begin{lemma}\label{lemma:chernoff}
Let $X_1, \dots, X_n$ be independent Bernoulli random variables with $\mathbb{P}(X_i = 1) = p_i$ for $i = 1, \dots, n$, such that $p_1 > p_2 > \dots > p_n$. Let $S_n = \sum_{i=1}^n X_i$ and $p = \sum_{i=1}^n p_i$, then
\[
\bbP(S_n > a n) \leq \left[\left(\frac{p_1}{a}\right)^a e^{a}\right]^n
\]
for $p_1 \leq a < 1$.
\end{lemma}
\begin{proof}
The proof follows the derivation of a similar Chernoff bound found in \citetSup{hagerup1990}. For $t \geq 0$, we have that
\begin{align*}
\bbP(S_n \geq an) &\leq e^{-tan} \bbE\left(e^{tS_n}\right) \\
&= e^{-tan} \prod_{i=1}^n (1 + p_i (e^t - 1)) \\
&\leq e^{-tan} (1 + p_1 (e^t - 1))^n,
\end{align*}
where the last line follows from the ordering of the $p_i$. For $t = \log(a(1-p_1)/[p_1(1-a)])$, this become
\begin{align*}
\bbP(S_n \geq an) &\leq \left[\left(\frac{p_1}{a}\right)^a \left(\frac{1 - p_1}{1 - a}\right)^{1 - a}\right]^n, \quad p_1 \leq a < 1.
\end{align*}
Note that
\[
\left(\frac{1 - p_1}{1 - a}\right)^{1 - a} = \left(1 + \frac{a - p_1}{1 - a}\right)^{1 -a} \leq e^{a - p_1} \leq e^{a}, \quad \text{ for } p_1 \leq a,
\]
so that
\[
\bbP(S_n \geq a n) \leq \left[\left(\frac{p_1}{a}\right)^a e^{a}\right]^n.
\]
Noting that $\bbP(S_n > an) \leq \bbP(S_n \geq an)$ gives the desired bound.
\end{proof}

\begin{lemma}\label{lemma:laplace}
    Assume $\blambda = (\lambda_1, \dots, \lambda_d)^{\top}$ is distributed as $\lambda_j \iidsim \Laplace(b)$ for $1 \leq j \leq d$. Then for any $\blambda_0 \in \Reals{d}$ and any $\epsilon > 0$,
\[
\bbP(\norm{\blambda - \blambda_0}_1 \leq \epsilon) \geq e^{-\frac{\norm{\blambda_0}_1}{b} - \frac{\epsilon}{b} - d \log(bd/\epsilon)}.
\]
\end{lemma}
\begin{proof}
Using the change-of-variables $\bu = \blambda - \blambda_0$, we get
\begin{align*}
\bbP(\norm{\blambda - \blambda_0}_1 \leq \epsilon) &= \int_{\norm{\blambda - \blambda_0}_1 \leq \epsilon} \left(\frac{1}{2b}\right)^d e^{-\norm{\blambda}_1/b} d\blambda \\
&\geq e^{-\norm{\blambda_0}_1/b} \int_{\norm{\bu}_1 \leq \epsilon} \left(\frac{1}{2b}\right)^de^{-\norm{\bu}_1/b} d\bu \\
&\geq e^{-\norm{\blambda_0}_1/b} \bbP\left(\sum_{i=1}^d E_i \leq \epsilon\right),
\end{align*}
where $E_1, \dots, E_d$ are iid exponential random variables with scale $b$. Recall that the sum of $n$ iid exponential random variables with scale $\theta$ follows a gamma distribution with shape $n$ and scale $\theta$. Thus, we have
\begin{align*}
\bbP\left(\sum_{i=1}^d E_i \leq \epsilon \right) &= \frac{1}{\Gamma(d) b^d} \int_0^{\epsilon} x^{d - 1} e^{-x/b} dx \\
&\geq \frac{1}{b^d (d-1)!} e^{-\epsilon/b} \int_0^{\epsilon} x^{d-1} = \frac{\epsilon^d}{b^d d!} e^{-\epsilon/b}.
\end{align*}
The fact that $d! \leq d^d$ completes the proof.
\end{proof}

\begin{lemma}\label{lemma:ss_laplace}
    Assume $\blambda = (\lambda_1, \dots, \lambda_d)^{\top}$ is distributed as
\begin{align*}
&\lambda_j \mid \xi_j \indsim \xi_j \Laplace(b) + (1 - \xi_j) \delta_0 \\
&\xi_j \indsim \Bern(\theta_j) 
\end{align*}
    for $\theta_1 > \theta_2 > \dots > \theta_d$ with $\theta_j \in (0, 1)$ for $1 \leq j \leq d$. Assume $\blambda_0 = (\lambda_{0,1}, \dots, \lambda_{0d})^{\top} \in \Reals{d}$ has nonzero support on the first $d_{0}$ entries, i.e., $\lambda_{0j} \neq 0$ for $j \leq d_0$ and $\lambda_{0j} = 0$ for $j > d_0$. Then for any $\epsilon > 0$,
\[
\bbP(\norm{\blambda - \blambda_0}_2 \leq \epsilon) \geq \theta_{d_0}^{d_0}(1 - \theta_1)^{d - d_0} \left[e^{-\norm{\blambda_0}_1/{b} - \epsilon/\sqrt{2}b} \left(\frac{\epsilon}{\sqrt{2}bd_0}\right)^{d_0}\right].
\]
\end{lemma}
\begin{proof}
We start with the inequality
\[
\bbP(\norm{\blambda - \blambda_0}_2 \leq \epsilon) \geq \left[\prod_{j > d_0} \bbP\left(\abs{\lambda_j} \leq \frac{\epsilon}{\sqrt{2 (d - d_0)}}\right)\right] \bbP\left(\norm{\lambda_{1:d_0} - \lambda_{0, 1:d_0}}_2 \leq \frac{\epsilon}{\sqrt{2}}\right).
\]
Note that $\bbP(\abs{\lambda_j} \leq \epsilon / \sqrt{2(d-d_0)}) \geq \bbP(\lambda_j = 0) = 1 - \theta_j \geq 1 - \theta_1$. Also,
\begin{align*}
\bbP\left(\norm{\lambda_{1:d_0} - \lambda_{0, 1:d_0}}_2 \leq \frac{\epsilon}{\sqrt{2}}\right) &\geq  \left[\prod_{j\leq d_0} \theta_j\right] \bbP_{\Laplace(b)}\left(\norm{\lambda_{1:d_0} - \lambda_{0, 1:d_0}}_2 \leq \frac{\epsilon}{\sqrt{2}}\right) \\
&\geq \theta_{d_0}^{d_0} \ \bbP_{\Laplace(b)}\left(\norm{\lambda_{1:d_0} - \lambda_{0, 1:d_0}}_1 \leq \frac{\epsilon}{\sqrt{2}}\right),
\end{align*}
where $\bbP_{\Laplace(b)}$ denotes the probability measure under a $\Laplace(b)$ density. Applying Lemma~\ref{lemma:laplace} completes the proof.
\end{proof}

\section{Full Conditional Distributions}\label{sec:full_cond}

In this section, we present the logarithm of the full conditional distribution used in the Hamiltonian Monte Carlo algorithm presented in Section~\ref{subsec:hmc_gibbs} of the main text. First, we outline the change-of-variable transformations that we applied to the scalar parameters with a constrained domain. We applied a log-transformation to $\sigma_{h}$ and a logistic transform to $\nu_h$ for $h = 1, \dots, d$. In particular, denote the unconstrained variables after the transformation as
\begin{align*}
\eta_{\sigma_h} &= \log \sigma_h, \qquad \eta_{\nu_h} = \log\left(\frac{\nu_h}{1 -\nu_h}\right), \qquad h = 1, \dots, d.
\end{align*}
    Note that we correct these transformations with Jacobian adjustments to the full conditional distribution. 
The logarithm of $p(\bpsi \mid Z_{1:d}, \bY_{1:T})$ up to an additive constant $C$ is
\begin{align}
    \log p(\bpsi \mid Z_{1:d}, \bY) = & \sum_{i \leq j} \frac{y_{ij} \, \left[\bTheta(Z_{1:d})\right]_{ij} + b(\left[\bTheta(Z_{1:d})\right]_{ij})}{\phi} + k(y_{ij}, \phi) \nonumber \\
&- \frac{\bbeta^{\top}\bbeta}{2 \sigma_{\beta}^2} - \frac{1}{2} \tr\left(\bB^{\top}\bB\right) - \sum_{h=1}^d \frac{\tilde{\lambda}_{h}^2}{2} - \sum_{h=1}^d \frac{e^{2 \eta_{\sigma_h}}}{2 b^2} \nonumber \\
&+\kappa \log\left[1 - \sigma(\eta_{\nu_1})\right]  + (\alpha - 1) \sum_{h=1}^d \log\left[\sigma(\eta_{\nu_h})\right]  \nonumber \\
&+ \sum_{h=1}^d \left[Z_h \log(\theta_h) + (1 - Z_h) \log(1 - \theta_h)\right] \nonumber \\
&+ \log \abs{\mathbf{J}(\eta_{\sigma_{1:d}}, \eta_{\nu_{1:d}})} + C, \label{eq:hmc_logpost}
\end{align}
where
\begin{equation*}
    \bTheta(Z_{1:d}) = (g \circ b')^{-1}\left(\sum_{k=1}^p  \bX_{k}\beta_k + \bU_{\bB} \bLambda(Z_{1:d}) \bU_{\bB}^{\top}\right).
\end{equation*}
The term involving the logarithm of the Jacobian determinant is
\begin{align*}
\log \abs{\mathbf{J}(\eta_{\sigma_{1:d}}, \eta_{\nu_{1:d}})} =  \sum_{h=1}^d \eta_{\sigma_h}  + \sum_{h=1}^d \log[\sigma(\eta_{\nu_h})(1 - \sigma(\eta_{\nu_h}))],
\end{align*}
where $\sigma(x) = 1/(1 + e^{-x})$ is the logistic function. 

\section{MCMC Initialization and Post-Processing}\label{sec:mcmc_info}

    This section provides additional information about the initialization of the MCMC algorithm proposed in Section~\ref{ss:sec:estimation} of the main text and the post-processing of its output.

\subsection{Parameter Initialization}\label{subsec:param_init}

 We can greatly reduce the number of iterations required to reach convergence by choosing good initial values for the model parameters. We initialize the parameters with the last sample from a short MCMC chain (500 iterations) targeting a model without dimension selection, that is, we repeat step one of Algorithm~\ref{alg:hmc_within_gibbs} for 500 iterations with $Z_h = 1$ for $1 \leq h \leq d$. We remove the dimension selection indicators to allow the sampler to initially freely explore the unconstrained parameter space before applying shrinkage. The initial parameters of this short chain are drawn randomly according to NumPyro's recommended random initialization scheme, which draws the unconstrained variables uniformly from the interval $(-2, 2)$~\citepSup{phan2019}. 

\subsection{Post-Processing to Account for Posterior Invariances}\label{subsec:post_invariance}

Although our use of an $\SSIBP$ prior places soft identifiability constraints on the parameters, the posterior remains invariant to signed-permutation of $\bU$'s columns, i.e., $\bU \rightarrow \bU \diag(\bs) \bP$ where $\bs \in \set{-1, 1}^d$ and $\bP$ is a permutation matrix on $d$ elements. This can lead to inflated variance estimates without proper post-processing. As such, we post-process the results by matching each posterior draw $\bU^{(s)}$ to some reference $\bU_0 \in \Reals{n \times d}$. For convenience, we choose $\bU_0$ to be the maximum a posteriori (MAP) estimator.

Similar to Procrustes matching commonly employed in the latent space literature, which minimizes the Frobenius norm over rotation matrices, we solve
\begin{align}
    \tilde{\bU}^{(s)} &= \argmin_{\bs \in \set{-1, 1}^d,\ \bP \in \Pi_d} \norm{\mathbf{U}_0 - \mathbf{U}^{(s)} \diag(\bs) \bP}_F^2 \nonumber \\
&=\argmax_{\bs \in \set{-1, 1}^d,\ \bP \in \Pi_d} \tr(\bP \, \mathbf{U}_0^{\top} \mathbf{U}^{(s)} \diag(\bs)). \label{eq:linear_assign}
\end{align}
where $\Pi_d$ is the set of permutation matrices on $d$ elements. We recognize the second expression as a linear assignment problem~\citepSup{gower2004} with a $d \times d$ cost matrix $C$ with elements
\[
    C_{ij} = \max\left(\mathbf{U}_{0,i}^{\top} \ \mathbf{U}_j^{(s)}, -\mathbf{U}_{0,i}^{\top} \ \mathbf{U}_j^{(s)}\right),
\]
where $\mathbf{U}_{0,i}$ and $\mathbf{U}_{j}^{(s)}$ are the $i$-th and $j$-th columns of $\mathbf{U}_0$ and $\mathbf{U}^{(s)}$, respectively. Therefore, the solution to 
(\ref{eq:linear_assign}) can be efficiently found in $O(d^3)$ time with algorithms such as the Hungarian method~\citepSup{kuhn1955}. In summary, for each posterior draw, we solve the above linear assignment problem to obtain $\bP$ and set $\bs = \text{sign}(\diag(\mathbf{U}_0^{\top} \mathbf{U}^{(s)} \bP))$. Note that we also apply the permutation matrix $\bP$ to the diagonal elements of $\bLambda^{(s)}$.

\subsection[A Point Estimator for the Latent Positions]{A Point Estimator for $\bU$}\label{subsec:point_estimates}

Because the posterior mean of $\bU$ may not lie on the Stiefel manifold, it does not provide a meaningful point estimator for $\bU$. To define a point estimator that lies on the Stiefel manifold, we take a Bayesian decision theoretic approach and minimize the following constrained Frobenius norm under the posterior:
\[
\argmin_{\bW \in \mathcal{V}_{d,n}} \E{ \norm{\bW - \bU}_F | \bY} = \argmax_{\bW \in \mathcal{V}_{d,n}} \tr(\E{\bU | \bY}^{\top} \bW).
\]
If we view the Stiefel manifold as an embedded manifold of $nd$-dimensional Euclidean space, then the solution is the posterior Fr\'echet mean of $\bU$. The above expression is minimized by the orthogonal component of the polar decomposition of $\E{\bU | \bY}$~\citepSup{gower2004}, i.e., $\bW = \bZ \bV^{\rm T}$ where $\E{\bU | \bY} = \bZ \bD \bV^{\top}$ is the singular value decomposition (SVD) of $\E{\bU | \bY}$.

\section{Definitions of Competing Model Selection Criteria}\label{sec:model_select}

First, we define the estimators of $d_0$ based on minimizing information criteria estimated with the posterior samples. Let $\set{\bxi^{(s)}}_{s=1}^S$ denote the $S$ posterior samples from the Bayesian model under consideration. In addition, let $\hat{\bxi}$ denote an estimate of the parameters based on these posterior samples. We use the definitions of the various information criteria as outlined in \citetSup{gelman2013}. Each criterion contains two terms, the first term measures model fit, while the second term penalizes model complexity. We define the criteria on the deviance scale, so that smaller values are better. The AIC is defined as
\[
\text{AIC} = -2 \log p(\bY \mid \hat{\bxi}) + 2 k,
\]
where $k$ is the number of model parameters, that is, $nd + d + p$. Similarly, the BIC is defined as
\[
    \text{BIC} = -2 \log p(\bY \mid \hat{\bxi}) + k \log{n \choose 2},
\]
where $k$ is the same as in the AIC. The DIC is defined as
\[
\text{DIC} = -2 \log p(\bY \mid \hat{\bxi}) + 2 p_{DIC},
\]
where $p_{DIC} = 2 \left(\log p(\bY \mid \hat{\bxi}) - \frac{1}{S} \sum_{s=1}^S \log p(\bY \mid \bxi^{(s)})\right)$. Next, WAIC approximates the leave-one-out cross-validation (LOO-CV) procedure with the expected log-predictive density as the loss function. For GLNEMs, we use the following formula for WAIC: 
\[
    \text{WAIC} = -2 \log \sum_{i < j} \left[\frac{1}{S} \sum_{s=1}^S p(Y_{ij} \mid \bxi^{(s)}) \right] + 2  \sum_{i < j} V_{s=1}^S \log p(\bY \mid \bxi^{(s)}),
\]
where $V_{s=1}^S \log p(\bY \mid \bxi^{(s)})$ denotes the sample variance of the log-likelihood evaluated over the posterior samples. This definition approximates leave-one-dyad-out cross-validation. Lastly, we considered the K-fold cross-validation estimator proposed by \citetSup{hoff2008}, which is outlined in Algorithm~\ref{alg:kfold}.
\begin{sampler}
\begin{enumerate}
\item Randomly split the set of dyads $\set{(i, j) \, : \, i < j}$ into $K$ test sets $A_1, \dots, A_K$.
\item For $d = 1, \dots, D$:
\begin{enumerate}
\item For $k = 1, \dots, K$:
\begin{enumerate}
\item[(1)] Perform the MCMC algorithm using only $\set{y_{ij} \, : \, (i, j) \notin A_k}$ and estimate the natural parameter matrices $\hat{\bTheta}$ for all dyads by their posterior means.
\item[(2)] Based on these estimates, compute the log predictive probability over all held-out dyads $L_d^{(k)} = \sum_{(i,j) \in A_{k}} \log \bbP(Y_{ij} = y_{ij} \mid \hat{\bTheta})$.
\end{enumerate}
\item Measure the predictive performance of $d$ as $L_d = \sum_{k=1}^K L_d^{(k)} / K$.
\end{enumerate}
\item Select $\hat{d} = \argmax_{1 \leq d \leq D} L_d$.
\end{enumerate}
\caption{K-Fold Cross-Validation for GLNEMs}
\label{alg:kfold}
\end{sampler}

\section{Sensitivity to Model Misspecification and Zero-Inflation}\label{sec:zero-inflation}

In this section, we provide a simulation study to assess the GLNEM's sensitivity to model misspecification caused by zero-inflation. In particular, we consider an ordinal-valued network generated from a zero-inflated Poisson model, that is,
\[
    Y_{ij} = Y_{ji} \indsim \pi \delta_0 + (1 - \pi) \Pois(\exp\left\{\bbeta^{\top}\bx_{ij} + [\bU \bLambda \bU^{\top}]_{ij}\right\}), \qquad 1 \leq i < j \leq n,
\]
where $\delta_0$ is a point mass at zero and $\pi \in (0, 1)$ is the probability that the edge variable is drawn from $\delta_0$. This model does not fall within the GLNEM framework; however, we posit that the negative binomial GLNEM with a log link may be able to recover the correct dimensionality of the latent space and certain model parameters. This hypothesis comes from the fact that the negative binomial GLNEM encompasses the over-dispersion implied by the zero-inflated Poisson model. Specifically, the zero-inflated Poisson model has the following mean and variance relationship:
\begin{align*}
    \mathbb{E}[Y_{ij} \mid \bx_{ij}] &= (1 - \pi) \mu_{ij},\\
    \textrm{Var}(Y_{ij} \mid \bx_{ij}) &= \mu_{ij} + [\pi/(1-\pi)] \mu_{ij}^2,
\end{align*}
where $\mu_{ij} = \exp\left\{\bbeta^{\top}\bx_{ij} + [\bU \bLambda \bU^{\top}]_{ij}\right\})$. As such, the negative binomial model correctly specifies the variance, but incorrectly specifies the mean. However, for the log link, this misspecification appears as constant shift of $\log(1 + \pi)$ in the intercept, which may only bias this parameter.

To empirically evaluate this hypothesis, we generated 50 networks from the zero-inflated Poisson model with zero-inflation probability $\pi = 0.1$ for $n = 100$ and 200. The model parameters and covariates were generated in the same way as the Poisson GLNEM outlined in Section~\ref{ss:subsec:competition} of the simulation study in the main text.  The true dimension of the latent space was $d_0 = 3$. We estimated a Poisson and negative binomial GLNEM with a log link and SS-IBP priors by running the MCMC algorithm proposed in Section~\ref{ss:sec:estimation} for 5,000 iterations after a burn-in of 5,000 iterations. We used the same hyperparamter settings as outlined in Section~\ref{ss:subsec:competition}.

\begin{figure}[htb]
    \centering
    \includegraphics[height=0.2\textheight, width=\textwidth, keepaspectratio]{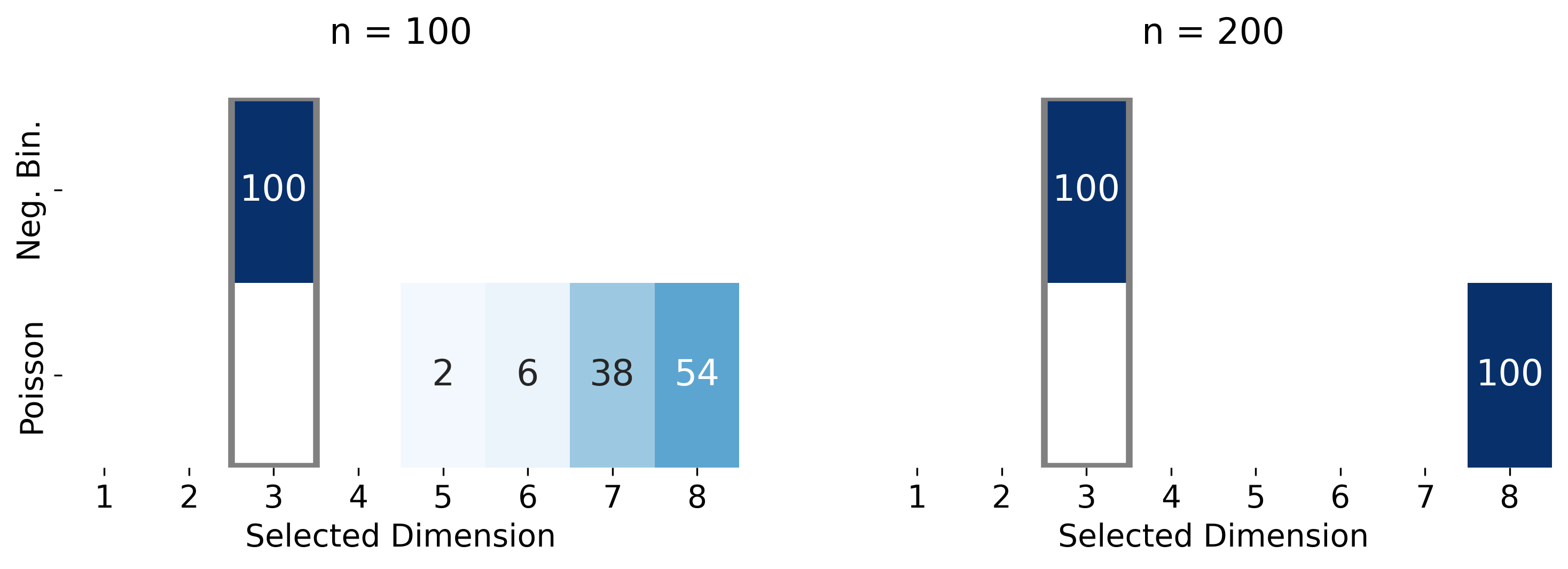}
    \caption{Heatmaps displaying the percentage of the 50 simulations in which the selected dimension was a particular value for the Poisson and negative binomial GLNEMs using the proposed SS-IBP prior. The column corresponding to the true value of $d_0 = 3$ is outlined in gray. Darker blue cells indicate percentages closer to $100\%$.} 
    \label{fig:heatmap_zip}
\end{figure}

Figure~\ref{fig:heatmap_zip} displays the results for estimating the latent space dimension. We see that the Poisson GLNEM drastically overestimates the dimension of the latent space for both sample sizes. In contrast, the negative binomial GLNEM correctly estimates the true dimension for all simulations. We believe this is due to the negative binomial GLNEM's ability to compensate for the zero-inflated Poisson model's over-dispersion through the extra dispersion parameter $\phi$. On the other hand, the Poisson GLNEM uses excessive dimensions to account for the over-dispersion due to zero-inflation.

Next, we examine the ability of the negative binomial GLNEM to recover the model parameters for this misspecified model. Table~\ref{tab:performance_zip} displays the trace correlation and relative errors for the model parameters estimated by the negative binomial GLNEM for both sample sizes. We see that all parameters except $\bbeta$ are accurately recovered and improve with $n$. The absolute errors of $\bbeta$'s elements, that is, the absolute value of the difference between the true and estimated values, are displayed in Table~\ref{tab:performance_zip_coef}. All coefficients except the intercept $\beta_1$ improve as $n$ increases. Also, the absolute error of $\beta_1$ is on the order of $\abs{\log(1-\pi)} = \abs{\log(0.9)} = 0.105$, which suggests that this bias is present in the intercept estimate. Overall, this simulation study suggests that one should prefer the negative binomial GLNEM when zero-inflation is suspected, which is often the case in count network data. 

\begin{table}[h]
\centering\small
\begin{tabular}{@{}llllll@{}} \toprule
$n$ & $\bU$ Trace Corr. & $\bLambda$ Rel. Error & $\bU \bLambda \bU^{\top}$ Rel. Error & $\bbeta$ Rel. Error \\
\bottomrule
100 & 0.986 (0.044) & 0.000587 (0.000588) & 0.033 (0.007) & 0.00618 (0.0035) \\
200 & 0.994 (0.056) & 0.000105 (0.000125) & 0.015 (0.002) & 0.00721 (0.0024) \\
\bottomrule
\end{tabular}
\caption{Trace correlation and relative estimation error of the model's parameters as the number of nodes ($n$) in the networks increased for the negative binomial GLNEM. Each cell contains the median value of the metric with the standard deviation over the 50 simulation displayed in parentheses.}
\label{tab:performance_zip}
\end{table}

\begin{table}[h]
\centering\small
\begin{tabular}{@{}cccccc@{}} \toprule
    $n$ & $\beta_1$ (Intercept) & $\beta_2$ & $\beta_3$ & $\beta_4$ & $\beta_5$ \\
\bottomrule
100 & 0.0814 (0.0285) & 0.0240 (0.0184) & 0.0190 (0.0191) & 0.0223 (0.0182) & 0.0164 (0.0132) \\
200 & 0.0994 (0.0185) & 0.0095 (0.0086) & 0.0115 (0.0093) & 0.0106 (0.0101) & 0.0094 (0.0089) \\
\bottomrule
\end{tabular}
    \caption{Absolute error of $\bbeta$'s elements as the number of nodes ($n$) in the networks increased for the negative binomial GLNEM. Each cell contains the median value of the metric with the standard deviation over the 50 simulation displayed in parentheses.}
\label{tab:performance_zip_coef}
\end{table}

\section{The Tweedie Distribution for Trade Data}\label{sec:tweedie}

In the main text, we stated that the Tweedie distribution is a suitable model for trade data due to various properties that we elaborate on in this section. The derivation of the following relationships can be found in \citetSup{dunn2005, dunn2008, dunn2018}. As in the real-data example in Section~\ref{subsec:banana}, we assume that we observe an adjacency matrix $\bY$ with entries $\set{Y_{ij} \, : \, 1 \leq i,j \leq n}$ that represent trade in thousands of U.S. Dollars occurring between two countries in a given time frame. The Tweedie GLNEM can be motivated by the following compound Poisson-gamma process. For each dyad, we assume an integer number of trades, $N_{ij} \geq 0$, occurred between country $i$ and country $j$ in a given time frame, where $N_{ij} \indsim \Pois(\lambda_{ij})$ with mean $\lambda_{ij}$. Note that $N_{ij}$ being zero indicates that no trade occurred between two countries in a given time frame. When $N_{ij} > 0$, we assume the amount in thousands of U.S. Dollars of each particular trade is $Z_{ij, k}$ for $k = 1, \dots, N_{ij}$, where $Z_{ij, k} \indsim \text{Gamma}(\alpha_{ij}, \beta_{ij})$ with shape parameter $\alpha_{ij}$ and scale parameter $\beta_{ij}$. The observed amount of trade $Y_{ij}$ between two nations is then
\[
    Y_{ij} = 
\begin{cases}
    \sum_{k=1}^{N_{ij}} Z_{ij, k}, & N_{ij} > 0, \\
    0, & N_{ij} = 0.
\end{cases}
\]
If we assume the network comes from a Tweedie GLNEM with mean $\mu_{ij} = \mathbb{E}[Y_{ij} \mid \bx_{ij}] = g^{-1}(\bbeta^{\top}\bx_{ij} + \bu_i^{\top}\bLambda\bu_j)$, dispersion $\phi$, and power parameter $1 < \xi < 2$, then we can relate the parameters to this underlying compound Poisson-gamma process as follows:
\begin{align*}
    \lambda_{ij} &= \frac{\mu_{ij}^{2 - \xi}}{\phi(2 - \xi)}, \\
    \alpha_{ij} &= \frac{2 - \xi}{\xi - 1}, \\
    \beta_{ij} &= \frac{\mu_{ij}^{\xi - 1}}{\phi (\xi - 1)}.
\end{align*}
Intuitively, the means of the underlying Poisson and gamma distributions increase as $\mu_{ij}$ increases. In terms of the latent positions, the Tweedie GLNEM assumes that as the similarity of two countries increases in latent space, the average number of trades they make and the average amount of each trade will also increase. Next, one can show that the Tweedie GLNEM has a power variance function $V(\mu_{ij}) = \mu_{ij}^{\xi}$, that is, the variance is heteroscedastic and increases as a power of the mean. \citetSup{silva2006} argued that this is a reasonable relationship for trade data.

Next, the probability of zero trade between two nations is
\[
    \mathbb{P}(Y_{ij} = 0 \mid \bx_{ij}) = \exp(-\lambda_{ij}) = \exp\left\{-\frac{\mu_{ij}^{2 - \xi}}{\phi (2 - \xi)}\right\}.
\]
From this relation, we see that the more dissimilar nations are in latent space, the more likely they are to have zero trade. In this way, the same covariates and latent variables are used to model the incidence of trade and the volume of trade, which is a model assumption that may not always be appropriate. 

Finally, the Tweedie distribution is scale-invariant in the sense that if $Y_{ij}$ comes from a Tweedie distribution with density $\text{Tw}(Y_{ij} \mid \mu_{ij}, \phi, \xi)$, then for any $c > 0$, $\tilde{Y}_{ij} = c Y_{ij}$ has a density $\text{Tw}(\tilde{Y}_{ij} | c\mu_{ij}, c^{2 - \xi} \phi, \xi)$, that is, Tweedie random variables are closed under scale transformations. This has important practical consequences for modeling trade using a log link function. When changing the units of measurement, e.g., from U.S. Dollars to Euros by multiplying by the exchange rate $c$, the mean will change to $c\mu_{ij} = \exp(\bbeta^{\top}\bx_{ij} + \bu_i^{\top}\bLambda \bu_j + \log c)$. That is, the intercept will be shifted, but the latent positions and covariate effects will remain the same, which is a desirable property for models of trade where the unit of measurement is arbitrary. This is not true for other models of trade, such as the tobit model, where the sign of the coefficient can flip under a scale change~\citepSup{head2014}.

\section{Additional Figures}\label{ss:subsec:additional_figures}

This section contains additional figures from the real data applications in Section~\ref{sec:applications} of the main text. In particular, this section includes cross-validation results, trace plots of the MCMC chains used to estimate the model parameters, and summaries of the latent space dimension's posterior distribution. 

Figure~\ref{fig:ppnetwork_kfold} contains the average held-out log-likelihoods for the fixed-dimension Bernoulli GLNEM estimated using K-fold cross-validation on the protein interaction network for $1 \leq d \leq 8$. The selected dimension based on the K-fold CV (1SE) estimator was $\hat{d}_0 = 3$.

\begin{figure}[ht]
\centering \includegraphics[width=\textwidth, height=0.25\textheight, keepaspectratio]{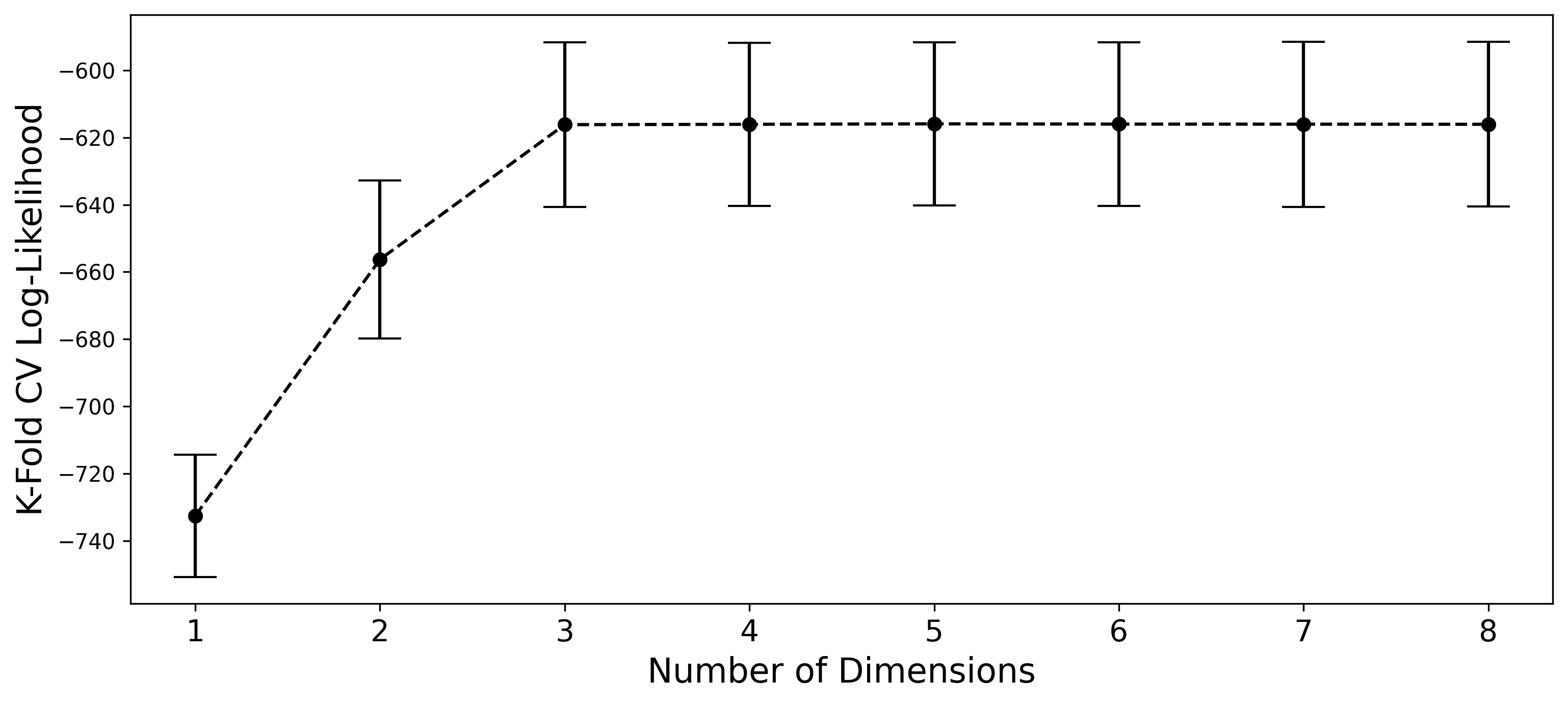}
    \caption{Dimension selection using K-Fold cross-validation with $K = 4$ for the Bernoulli GLNEM with a $N(\mathbf{0}_d, n \mathbf{I}_d)$ prior for $\bLambda$ on the protein interaction network. Error bars indicate one standard errors calculated over the four folds.}
\label{fig:ppnetwork_kfold}
\end{figure}

Figure~\ref{fig:pp_network_mcmc} contains the diagnostic plots for the Bernoulli GLNEM estimated on the network of protein interactions in Section~\ref{subsec:ppnetwork} of the main text. The plots contain only the 15,000 samples after burn-in. All trace plots look well mixed. Furthermore, we see the posterior distribution places 84\% of its mass on the 5 dimensional model. 

Figure~\ref{fig:tree_poisson_mcmc} contains the diagnostic plots for the Poisson GLNEM estimated on the tree network analyzed in Section~\ref{subsec:tree} of the main text. This model frequently switches between a two and three dimension model, see the trace plot on the lower left. As a result the trace plots for the coefficients appear less stationary; however, this is an artifact of the dimension switching. When the model visits a dimension for a significant period of time, the chains for the model coefficients do appear to mix relatively well.

Figure~\ref{fig:tree_negbinom_mcmc} contains the diagnostic plots for the negative binomial GLNEM estimated on the tree network analyzed in Section~\ref{subsec:tree} of the main text. This posterior is largely concentrated on a single latent space dimension. The trace plots indicate that the model parameters are well mixed.

Lastly, Figure~\ref{fig:banana_tweedie_mcmc} contains the diagnostic plots for the Tweedie GLNEM estimated on the banana trade network analyzed in Section~\ref{subsec:banana} of the main text. This posterior contains a mixture of five, six, and seven dimensions. The trace plots appear relatively well mixed.

\begin{figure}[htp]
    \centering \includegraphics[width=\textwidth, keepaspectratio]{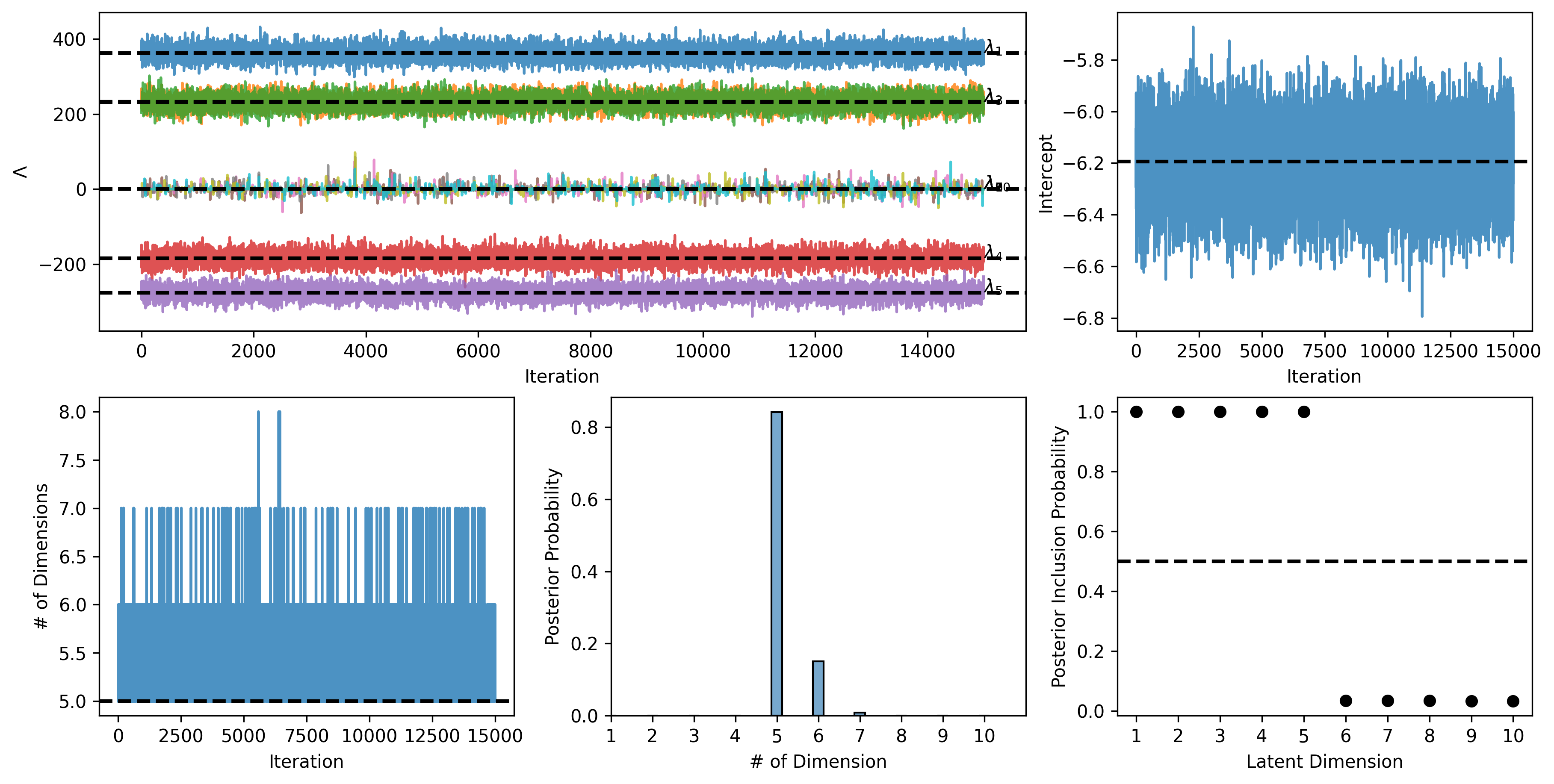}
    \caption{Diagnostic plots for the Bernoulli GLNEM estimated on the protein interaction network. (Top) Trace plots for certain model parameters: (Left) Trace plots of the diagonal elements of $\bLambda$, and (Right) trace plot of the intercept. (Bottom) Summaries of the number of latent dimensions inferred by the SS-IBP prior: (Left) The trace plot of the number of dimensions visited by the Markov chain, (Center) the posterior probabilities over the number of dimensions, and (Right) the posterior inclusion probabilities of a given dimension.}
\label{fig:pp_network_mcmc}
\end{figure}

\begin{figure}[htp]
    \centering \includegraphics[width=\textwidth, keepaspectratio]{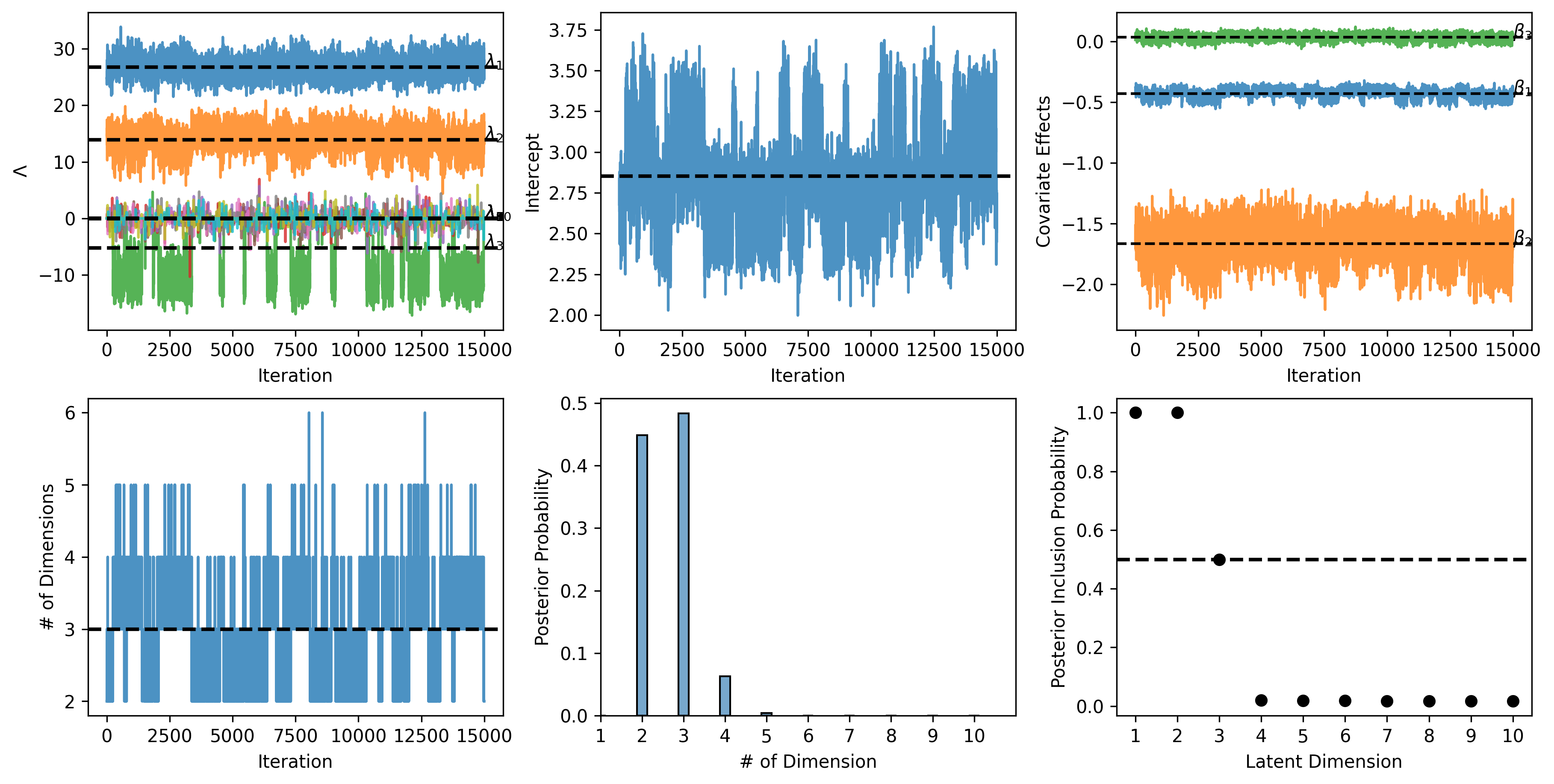}
    \caption{Diagnostic plots for the Poisson GLNEM estimated on the host-parasite interaction network. (Top) Trace plots for certain model parameters: (Left) Trace plots of the diagonal elements of $\bLambda$, (Center) trace plot of the intercept, and (Right) trace plots of dyadic covariates' coefficients. (Bottom) Summaries of the number of latent dimensions inferred by the SS-IBP prior: (Left) The trace plot of the number of dimensions visited by the Markov chain, (Center) the posterior probabilities over the number of dimensions, and (Right) the posterior inclusion probabilities of a given dimension.}
    \label{fig:tree_poisson_mcmc}
\end{figure}

\begin{figure}[htp]
    \centering \includegraphics[width=\textwidth, keepaspectratio]{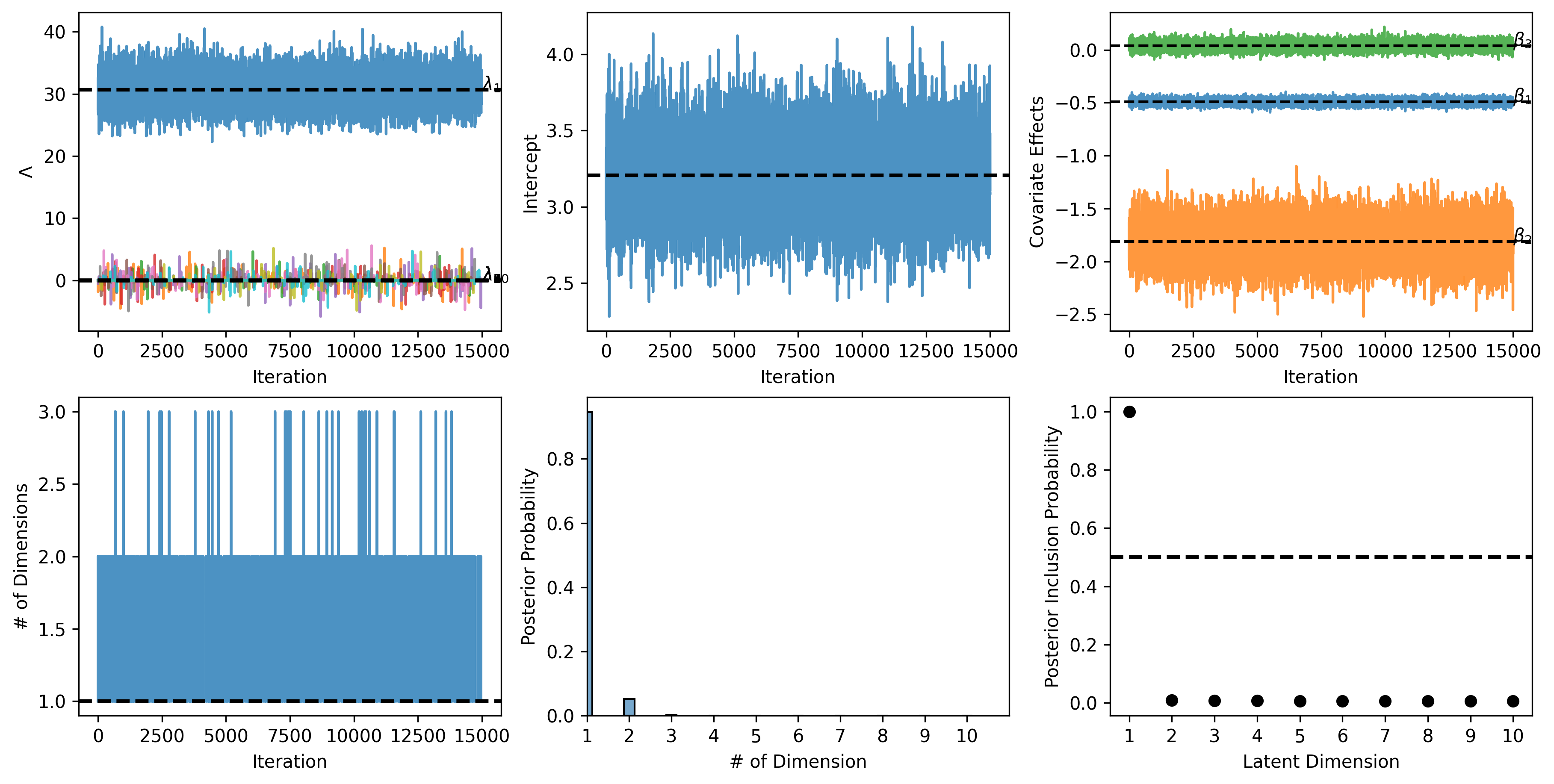}
    \caption{Diagnostic plots for the negative binomial GLNEM estimated on the host-parasite interaction network. (Top) Trace plots for certain model parameters: (Left) Trace plots of the diagonal elements of $\bLambda$, (Center) trace plot of the intercept, and (Right) trace plots of dyadic covariates' coefficients. (Bottom) Summaries of the number of latent dimensions inferred by the SS-IBP prior: (Left) The trace plot of the number of dimensions visited by the Markov chain, (Center) the posterior probabilities over the number of dimensions, and (Right) the posterior inclusion probabilities of a given dimension.}
    \label{fig:tree_negbinom_mcmc}
\end{figure}

\begin{figure}[htp]
    \centering \includegraphics[width=\textwidth, keepaspectratio]{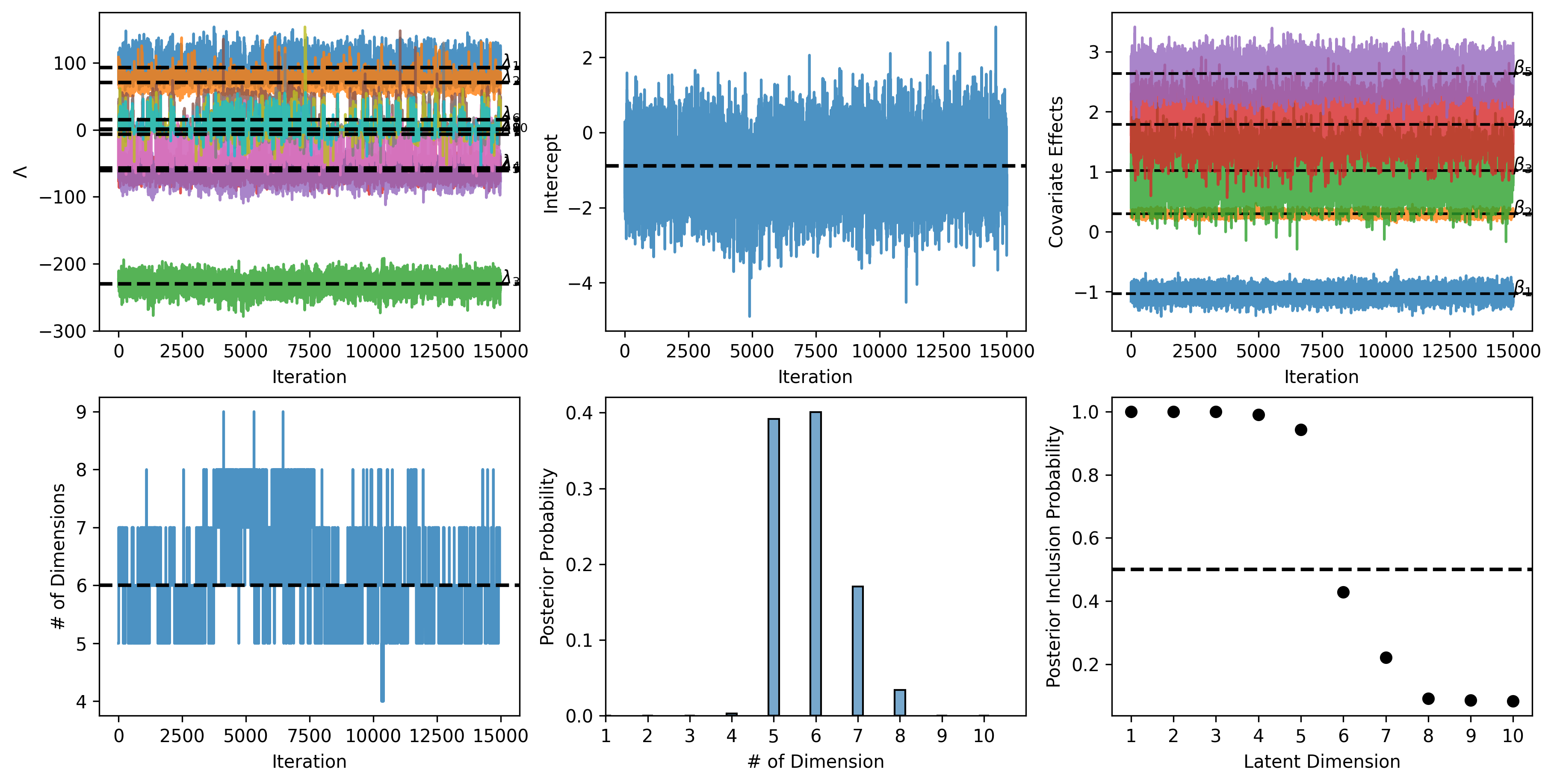}
    \caption{Diagnostic plots for the Tweedie GLNEM estimated on the banana trade network. (Top) Trace plots for certain model parameters: (Left) Trace plots of the diagonal elements of $\bLambda$, (Center) trace plot of the intercept, and (Right) trace plots of dyadic covariates' coefficients. (Bottom) Summaries of the number of latent dimensions inferred by the SS-IBP prior: (Left) The trace plot of the number of dimensions visited by the Markov chain, (Center) the posterior probabilities over the number of dimensions, and (Right) the posterior inclusion probabilities of a given dimension.}
    \label{fig:banana_tweedie_mcmc}
\end{figure}

\clearpage

\bibliographystyleSup{asa}
\bibliographySup{references}

\end{appendix}

\end{document}